\documentclass{article}
\usepackage{ijcai23}

\usepackage{times}
\usepackage{soul}
\usepackage{url}
\usepackage[hidelinks]{hyperref}
\usepackage[utf8]{inputenc}
\usepackage[small]{caption}
\usepackage{graphicx,color}
\usepackage{amsmath, bm}
\usepackage{amssymb}
\usepackage{amsfonts}
\usepackage{amsthm}
\usepackage{dsfont}
\usepackage{booktabs}
\usepackage[switch]{lineno}

\usepackage{bbm}
\usepackage{nicefrac}

\usepackage[noend]{algorithm2e}

\RestyleAlgo{ruled}
\usepackage{mathtools}
\mathtoolsset{showonlyrefs}
\newtheorem{example}{Example}
\newtheorem{theorem}{Theorem}
\newtheorem{definition}{Definition}
\newtheorem{assumption}{Assumption}
\newtheorem{lemma}{Lemma}
\newtheorem{problem}{Problem}

\usepackage{pgfplots}
\DeclareUnicodeCharacter{2212}{−}
\usepgfplotslibrary{groupplots, dateplot}
\usetikzlibrary{patterns, shapes, automata, arrows, shapes.arrows, positioning, decorations.pathreplacing, calligraphy, calc}

\pgfplotsset{compat=newest}

\usepackage{tikz}

\usepgfplotslibrary{external}
\pgfplotsset{compat=newest}
\newenvironment{customlegend}[1][]{%
	\begingroup
	\csname pgfplots@init@cleared@structures\endcsname
	\pgfplotsset{#1}%
}{%
	\csname pgfplots@createlegend\endcsname
	\endgroup
}%

\def\addlegendimage{\csname pgfplots@addlegendimage\endcsname}
\newcommand\blfootnote[1]{%
  \begingroup
  \renewcommand\thefootnote{}\footnote{#1}%
  \addtocounter{footnote}{-1}%
  \endgroup
}

\title{Differential Privacy in Cooperative Multiagent Planning}

\author{
Bo Chen$^{1,*}$
\and
Calvin Hawkins$^{1,*}$\and
Mustafa O. Karabag$^{2,*}$\and
Cyrus Neary$^{2,*}$\and\\
Matthew Hale$^{1}$\And
Ufuk Topcu$^{2}$
\affiliations
$^1$The University of Florida\\
$^2$The University of Texas at Austin
\emails
\{bo.chen, calvin.hawkins, matthewhale\}@ufl.edu,
\{karabag, cneary, utopcu\}@utexas.edu
}
\begin{document}

% Colors
\definecolor{minDependencyPolicy}{RGB}{180, 65, 161}
\definecolor{baselinePolicy}{RGB}{46, 36, 47}

%Math symbols
\newcommand{\expectation}{\mathbb{E}}
\newcommand{\kl}{KL}
\newcommand{\entropy}{H}
\newcommand{\distribution}{\Delta}
\newcommand{\probabilityMeasure}{\mu}
\newcommand{\genericRandomVar}{Y}
\newcommand{\genericRandomVarSupport}{\mathcal{\genericRandomVar}}
\newcommand{\genericDistribution}{Q}
\newcommand{\genericDistributionSupport}{\mathcal{\genericDistribution}}
\newcommand{\genericFunction}{f}
\newcommand{\genericFunctionAlt}{g}
\newcommand{\emptyString}{\varepsilon}

\newcommand{\epsilonTransition}{\alpha}

\newcommand{\genericString}{w}
\newcommand{\constantNumber}{K}
\newcommand{\genericSet}{V}

% MDP
\newcommand{\mdp}{\mathcal{M}}
\newcommand{\mdpState}{s}
\newcommand{\mdpStateAlt}{y}
\newcommand{\mdpInitialState}{\mdpState_{I}}
\newcommand{\mdpStateSpace}{\mathcal{S}}
\newcommand{\mdpAction}{a}
\newcommand{\mdpActionAlt}{b}
\newcommand{\mdpActionSpace}{\mathcal{A}}
\newcommand{\mdpReward}{\mathcal{R}}
\newcommand{\mdpTransition}{\mathcal{T}}
\newcommand{\outdegree}{\rho}

\newcommand{\policy}{\pi}

\newcommand{\mdpPath}{\xi}
\newcommand{\mdpStateSeq}{h}
\newcommand{\mdpPathDist}{\Gamma}
\newcommand{\mdpValue}{v}
\newcommand{\mdpStateActionProcess}{X}
\newcommand{\mdpJointPathProcess}{\bm{X}}
\newcommand{\mdpStationaryStateActionProcess}{\bar{X}}
\newcommand{\mdpStateActionProcessAlt}{Y}
\newcommand{\mdpMixedStateActionProcess}{\bar{\mdpStateActionProcess}}
\newcommand{\mdpStateRandomVar}{S}
\newcommand{\mdpActionRandomVar}{A}

\newcommand{\timeHorizon}{T}
\newcommand{\randomReachTime}{\eta}

% types of policies
%%%%%%%%% CHANGE THESE %%%%%%%%%%%%%%
\newcommand{\joint}{joint}
\newcommand{\fullcommunication}{full}
\newcommand{\fullyimaginary}{full\text{ }img}
\newcommand{\imaginary}{img}
\newcommand{\intermittent}{int}

% Markov game
\newcommand{\game}{\bm{\mdp}}
\newcommand{\gameState}{\bm{\mdpState}}
\newcommand{\gameStateAlt}{\bm{\mdpStateAlt}}
\newcommand{\gameActionAlt}{\bm{\mdpActionAlt}}
\newcommand{\gameInitialState}{\bm{\mdpInitialState}}
\newcommand{\gameStateSpace}{\bm{\mdpStateSpace}}
\newcommand{\gameAction}{\bm{\mdpAction}}
\newcommand{\gameActionSpace}{\bm{\mdpActionSpace}}
\newcommand{\gameTransition}{\bm{\mdpTransition}}
\newcommand{\gameReward}{\bm{\mdpReward}}
\newcommand{\gameStateRandomVar}{\bm{\mdpStateRandomVar}}
\newcommand{\gameActionRandomVar}{\bm{\mdpActionRandomVar}}
\newcommand{\gameStateActionProcessAlt}{\bm{\mdpStateActionProcessAlt}}
\newcommand{\len}{len}
\newcommand{\expectedLength}{l}
\newcommand{\pathSet}{W}
\newcommand{\reachPathSet}{R}
\newcommand{\gameAbsorbingState}{\gameState_\alpha}

\newcommand{\jointPolicy}{\bm{\policy}}
\newcommand{\localPolicy}{\policy}

\newcommand{\targetSet}{\gameStateSpace_{\mathcal{T}}}

\newcommand{\deadSet}{\gameStateSpace_{\mathcal{A}}}
\newcommand{\deadSetPrime}{\gameStateSpace_{\mathcal{D}}}
\newcommand{\doneSet}{\gameStateSpace_{\mathcal{E}}}
\newcommand{\gameProcessEndState}{\gameState_{\epsilonTransition}}
\newcommand{\mdpProcessEndState}{\mdpState_{\epsilonTransition}}

\newcommand{\gamePath}{\bm{\mdpPath}}
\newcommand{\gamePathDist}{\bm{\mdpPathDist}}
\newcommand{\gameValue}{\bm{\mdpValue}}
\newcommand{\gameStateActionProcess}{\bm{\mdpStateActionProcess}}
\newcommand{\totalCorrelation}{C}
\newcommand{\totalCorrelationUpperBound}{\bar{\totalCorrelation}}

% multi-agent macros
\newcommand{\numAgents}{N}

%Communication systems
\newcommand{\probabilityFailureForever}{p}
\newcommand{\probabilityFailureOneStep}{q}
\newcommand{\sequenceCommAvailibility}{\Lambda}
\newcommand{\oneStepCommAvailibility}{\lambda}

% Optimization problem macros
\newcommand{\expectedLengthCoef}{\delta}
\newcommand{\totalCorrelationCoef}{\beta}
\newcommand{\occupancyVar}{x}

% Running example
\newcommand{\rover}{R}
\newcommand{\robot}{R}
\newcommand{\agent}{R}
\newcommand{\base}{B}
\newcommand{\goal}{T}

% Privacy macros
\newcommand{\adj}{\textnormal{Adj}}
\newcommand{\oblivious}{tr}
\newcommand{\private}{pr}
\newcommand{\genericPath}{\bm{w}}
\newcommand{\privatePath}{\Tilde{\bm{w}}}
\newcommand{\privateMdpState}{\Tilde{\mdpState}}
\newcommand{\privateGameState}{\Tilde{\gameState}}
\newcommand{\privacyLevel}{\epsilon}
\newcommand{\adjParam}{k}
\newcommand{\trueStateProb}{\tau}

\newcommand{\probONB}{p_{\textrm{onb}}}
\newcommand{\probOffline}{p_{\textrm{off}}}
\newcommand{\probRepair}{p_{\textrm{r}}}

\maketitle

\begin{abstract}

Privacy-aware multiagent systems must protect agents' sensitive data while simultaneously ensuring that agents accomplish their shared objectives. Towards this goal, we propose a framework to privatize inter-agent communications in cooperative multiagent decision-making problems. We study sequential decision-making problems formulated as cooperative Markov games with reach-avoid objectives. We apply a differential privacy mechanism to privatize agents' communicated symbolic state trajectories, and then we analyze tradeoffs between the strength of privacy and the team's performance. For a given level of privacy, this tradeoff is shown to depend critically upon the total correlation among agents' state-action processes. We synthesize policies that are robust to privacy by reducing the value of the total correlation. Numerical experiments demonstrate that the team's performance under these policies decreases by only \(3\) percent when comparing private versus non-private implementations of communication. By contrast, the team's performance decreases by roughly \(86\) percent when using baseline policies that ignore total correlation and only optimize team performance.
\blfootnote{\textsuperscript{*} Indicates equal contribution.}

\end{abstract}

\section{Introduction}

In cooperative multiagent systems, a team of decision-making agents interact with a shared environment to accomplish a common objective~\cite{Cao2013distributed,Parker2016robot}.
In these systems, inter-agent communication is often necessary for the successful coordination of the team; each agent typically relies on information pertaining to its teammates while making its own decisions.
However, this communicated information may be sensitive.
For example, it might be beneficial for autonomous vehicles to share location data while solving multi-vehicle routing problems. 
However, this data would reveal the passengers' sensitive location data.
Privacy-aware multiagent systems should thus protect the agents' sensitive data, while simultaneously ensuring that the agents are able to accomplish their common objective.

In this work, we develop such privacy-aware multiagent systems.
In particular, we study sequential multiagent decision problems formulated as cooperative Markov games with reach-avoid objectives.
We assume that a trusted central aggregator is used to synthesize a collection of local policies for the team of agents \emph{a priori}. 
A local policy of an agent is a mapping from the joint state space of the agents to the agent's local action space. 
However, during policy execution, the agents want to keep their individual state trajectories private from their teammates and from potential eavesdroppers (the aggregator is not involved at run time). 
When the local policies do not take privacy into consideration, their performance under private communications can decrease dramatically, as shown by our numerical results.
Thus, we develop a framework to privatize the inter-agent communications required to execute the policies, and to synthesize policies that are performant under private communications.

We use \textit{differential privacy}~\cite{dwork2014algorithmic} to develop a framework providing formal privacy guarantees in multiagent systems.
In the Markov game, each agent is modeled by a Markov decision process (MDP) and we are concerned with privatizing the state trajectories of these MDPs. 
We implement differential privacy using the Online Mechanism for Markov chains presented in~\cite{chen2022differential}. 
This mechanism guarantees differential privacy for the symbolic state trajectories produced by MDPs, provides an efficient method for agents to generate private states in real time, and ensures that the private trajectory is feasible with respect to the underlying dynamics of the MDP.
The strength of these privacy guarantees can be tuned by each agent. 

Our specific contributions in this work are as follows:

\begin{enumerate}

    \item \textit{A framework for differential privacy in multiagent systems.} We propose a framework for differential privacy in multiagent planning problems.
    The framework allows for the decentralized execution of local policies under private inter-agent communications.
    
    \item \textit{Theoretical results: Analyzing the tradeoff between privacy and performance.} 
    We bound the team's success probability under private communications in terms of the strength of privacy and total correlation of agents' state-action processes.
    
    \item \textit{Synthesis of policies to balance privacy and performance.} 
    By minimizing this total correlation value, we use the tradeoffs between privacy and performance to synthesize policies for the multiagent system that achieve high performance under strong levels of privacy.
    
\end{enumerate}

Numerical experiments demonstrate the strong performance of the synthesized policies, even with private communications.
We observe that under privatized communication: 1) the proposed minimum-dependency policies are \(80\) percent more performant than baseline policies that only optimize the team's performance under truthful communications and that ignore total correlation, 2) as the total correlation decreases, the team's performance increases, and 3) the performance of the minimum-dependency policies is robust to the level of privacy enforced by the privacy mechanism.

Despite the importance of privacy in multiagent systems~\cite{such2014privacysurvay}, existing algorithms for multiagent planning and learning typically do not examine the tradeoff between privacy and team performance, and many do not consider privacy at all.
\cite{nissim2014distributed,brafman2015privacy} introduce the notion of \textit{strong privacy} in mutliagent planning for deterministic problems. These works develop algorithms that ensure agents do not share sensitive states or actions when executing a distributed planning algorithm. \cite{Ye2022} replaces the notion of strong privacy with differential privacy to privatize the information needed for decentralized planning in a deterministic case. \cite{hefner2022privacy} extends the notion of strong privacy to stochastic systems and develops a distributed value iteration algorithm. 
These works are concerned with hiding a private portion of each agent's states, and they do not consider mechanisms in which the agents achieve privacy by altering their shared information.
By contrast, our work studies a differential privacy mechanism that alters the state trajectories of the agents during multiagent communication in stochastic environments.

Meanwhile, differential privacy has been studied in the context of planning and reinforcement learning for MDPs \cite{garcelon2020local,qiao2022offline,gohari2021differential}. 
However, these works study single-agent problems and they are mainly concerned with privatizing value functions, reward values, or transition probabilities.
Our work instead considers the multiagent setting and we define differential privacy over symbolic state trajectories.
In particular, we extend the differential privacy mechanism presented in \cite{chen2022differential} to multiagent planning problems, and we study the impact of privacy on the team's performance.

Decentralized policy execution has gained attention for planning and reinforcement learning in multiagent MDPs \cite{becker2003transition,rashid2018qmix,son2019qtran,oliehoek2016concise,Karabag2022}.
As a byproduct of decentralized policy execution, these algorithms may achieve privacy in the sense that agents do not communicate locally available information.
However, these works do not explicitly consider privacy or give privacy guarantees.
We instead allow for communication and use total correlation as a soft decentralization metric, which enables the synthesis of policies that are performant under private communications.

\section{Preliminaries}
\label{sec:prelims}

The entropy of a discrete random variable $Y$ with a support $\mathcal{Y}$ is $H(Y)=-\sum_{y \in \mathcal{Y}} \operatorname{Pr}(Y=y) \log (\operatorname{Pr}(Y=y))$. 

\subsection{Cooperative Markov Games }
\label{sec:prelim_games}
Given a finite collection of $\numAgents$ agents indexed by $i\in\{1,2,\dots,\numAgents\},$ we model the dynamics of agent $i$ with an MDP $\mdp^i$. An MDP is a tuple $\mdp^i = (\mdpStateSpace^i,\mdpInitialState^i,\mdpActionSpace^i,\mdpTransition^i),$ where $\mdpStateSpace^i$ is agent $i$'s finite set of local states, $\mdpInitialState^i\in\mdpStateSpace^i$ is an initial state, $\mdpActionSpace^i$ is agent $i$'s finite set of local actions, and $\mdpTransition^i:\mdpStateSpace^i\times\mdpActionSpace^i\to\distribution(\mdpStateSpace^i)$ is a transition probability function, where $\distribution(\mdpStateSpace^i)$ denotes the set of probability distributions over the state space $\mdpStateSpace^i.$ For brevity, we use $\mdpTransition^i(s^i,a^i,y^i)$ to denote the probability of $y^i$ given by the distribution $\mdpTransition^i(s^i,a^i).$ A state $\mdpState^i_j\in\mdpStateSpace^i$ is called a \emph{feasible state} of another state $\mdpState^i_k\in\mdpStateSpace^i$ if there exists an action $a^i\in\mdpActionSpace^i$ such that $\mdpTransition^i(\mdpState^i_k,\mdpAction^i,\mdpState^i_j)>0$. 

Given such a collection of agents, we formulate the team's decision problem as a cooperative Markov game $\game.$ A cooperative Markov
game involving $\numAgents$ agents, each of which is modeled by an MDP $\mdp^i = (\mdpStateSpace^i,\mdpInitialState^i,\mdpActionSpace^i,\mdpTransition^i),$ is given by the tuple $\game = (\gameStateSpace,\gameInitialState,\gameActionSpace,\gameTransition).$ Here, $\gameStateSpace=\mdpStateSpace^1\times\dots\times\mdpStateSpace^\numAgents$ is the joint state space, $\gameInitialState=(\mdpInitialState^1,\dots,\mdpInitialState^\numAgents)$ is the joint initial state, $\gameActionSpace=\mdpActionSpace^1\times\dots\times\mdpActionSpace^\numAgents$ is the joint action space, and $\gameTransition$ is the joint transition probability function. For brevity, we use $\gameTransition(\gameState,\gameAction,\bm{y})$ to denote the probability of $\bm{y}$ given the distribution $\gameTransition(\gameState,\gameAction).$ Let $\gameState\in\gameStateSpace$ and $\gameAction\in\gameActionSpace$ denote a joint state and action, respectively. $\gameTransition$ is defined as $\gameTransition(\gameState,\gameAction,\gameStateAlt)=\prod_{i=1}^\numAgents \mdpTransition^i(\mdpState^{i},\mdpAction^{i},\mdpStateAlt^{i})$ for all $\gameState=(\mdpState^1,\dots,\mdpState^\numAgents)\in\gameStateSpace,$ $\gameStateAlt=(\mdpStateAlt^1,\dots,\mdpStateAlt^\numAgents)\in\gameStateSpace$ and $\gameAction=(\mdpAction^1,\dots,\mdpAction^\numAgents)\in \gameActionSpace.$

For notational convenience, we use $\gameState^{-i} \in \mdpStateSpace^1 \times \ldots \times \mdpStateSpace^{i-1} \times \mdpStateSpace^{i+1} \times$ $\ldots \times \mdpStateSpace^\numAgents$ to denote the states of agent $i$ 's teammates, excluding agent $i$ itself. By $\gameStateSpace^{-i}=\mdpStateSpace^1 \times \ldots \times \mdpStateSpace^{i-1} \times \mdpStateSpace^{i+1} \times \ldots \mdpStateSpace^\numAgents$, we denote the values $\gameState^{-i}$ can take. Similarly $\gameAction^{-i}$ and $\gameActionSpace^{-i}$ denote the actions of agent $i$ 's teammates and the set of all possible actions of teammates, respectively.

A (stationary) local policy \(\localPolicy^{i} : \gameStateSpace \to \distribution(\mdpActionSpace^{i})\) of Agent \(i\) is a mapping from a particular joint state to a probability distribution over actions of Agent \(i\). Given the team is in joint state \(\gameState\), \(\localPolicy^{i}(\gameState, \mdpAction^{i})\) denotes the probability that action \(\mdpAction^{i}\) is selected by \(\localPolicy^{i}\) for agent \(i\) . We define a (stationary) joint policy \(\jointPolicy\) to be a collection of local policies, \(\lbrace \localPolicy^{i} \rbrace_{i=1}^{\numAgents}\). 

In a truthful communication setting, at each timestep each agent \(i\) observes its local state \(\mdpState^{i}_{t}\), and communicates this information with all of its teammates. 
Each agent then uses the information communicated by its teammates to construct the team's joint state \(\gameState \in \gameStateSpace\), and subsequently, it uses its local policy \(\policy^{i}(\gameState)\) to sample an action \(\mdpAction^{i} \in \mdpActionSpace^{i}\) to execute. 

In this work we consider team reach-avoid problems. 
That is, the team's objective is to collectively reach 
a target set \(\targetSet \subseteq \gameStateSpace\) of states, while avoiding a set \(\deadSet \subseteq \gameStateSpace\) of states.
The centralized planning problem then is to solve for a collection of local policies \(\lbrace \localPolicy^{i} \rbrace_{i=1}^{\numAgents}\) maximizing the probability of reaching \(\targetSet\) from the team's initial joint state \(\gameInitialState\), while avoiding \(\deadSet\).
We call this probability value the success probability.
More formally, we say that a state-action \textit{trajectory} \(\gamePath = \gameState_0 \gameAction_0 \gameState_1 \gameAction_1 \ldots\) successfully reaches the target set \(\targetSet\) if there exists some time \(M\) such that \(\gameState_{M} \in \targetSet\) and for all \(t < M\), \(\gameState_t \not \in \deadSet\).
While we focus on reach-avoid problems, our framework can be applied to settings with generic rewards.

We use $\occupancyVar_{\gameState, \gameAction}$ to denote the occupancy measure of the state-action pair $(\gameState, \gameAction)$, i.e., the expected number of times that action $\gameAction$ is taken at state $\gameState$. 
Similarly, $\occupancyVar_{\mdpState^i, \mdpAction^i}$ denotes the the occupancy measure of the state-action pair $\left(\mdpState^i, \mdpAction^i\right)$ for agent $i$ where $\occupancyVar_{\mdpState^i, \mdpAction^i}=$ $\sum_{\bm{s^{-i}} \in \gameStateSpace^{-i}} \sum_{\gameAction^{-i} \in \gameActionSpace^{-i}} \occupancyVar_{(\mdpState^i, \gameState^{-i}), (\mdpAction^{i}, \gameAction^{-i})}.$ 
Let \( \deadSetPrime\) be the states from which the probability of reaching \(\targetSet\) is \(0\) under any collection of local policies. The following assumption ensures that every trajectory satisfies or violates the reachability specification in finite time. 
\begin{assumption}
    The total occupancy measure is finite at states \(\gameStateSpace \setminus (\targetSet \cup \deadSetPrime)\), i.e., \(\sum_{\gameState \in \gameStateSpace\setminus(\targetSet \cup \deadSetPrime), \gameAction \in \gameActionSpace} \occupancyVar_{\gameState, \gameAction} < \infty\).
\end{assumption}

A state-action trajectory $\mdpPath^i$ of the MDP $\mathcal{M}^i$ is a sequence $\mdpPath^i=\mdpState_0^i \mdpAction_0^i \mdpState_1^i \mdpAction_1^i\dots$ such that for all $t=0,1,\dots,$ $\mdpTransition(\mdpState_t^i,\mdpAction_t^i,\mdpState_{t+1}^i)>0.$ We use $\gamePath=\gameState_{0}\gameAction_{0}\gameState_{1}\dots$ to denote the joint state-action trajectory of all agents and $\gamePath^{-i}=\gameState_{0}^{-i}\gameAction_{0}^{-i}\gameState_{1}^{-i}\gameAction_{1}^{-i}\dots$ to denote joint state-action trajectory with agent $i$ excluded. Note that \(\gamePath\) and \(\gamePath^{-i}\) are both strings of vectors. We define the effective length of strings \(\len(\gamePath = \gameState_0 \gameAction_0\ldots) = \min\lbrace t+1 | \gameState_{t} \in \targetSet \cup \deadSetPrime \rbrace \). 
Let agent $i$'s state trajectory up to time $t$ be $\mdpStateSeq_t^i=\mdpState_0^i \mdpState_1^i \dots \mdpState_t^i.$ We are concerned with the privacy of $\mdpStateSeq_t^i$ so that agents can execute their policy without revealing sensitive information. 

\subsection{Differential Privacy}
Differential privacy is enforced by a \emph{mechanism}, which is a randomized map. We enforce differential privacy on a per-agent basis, an approach sometimes called ``local differential privacy". For nearby local state trajectories, a mechanism must produce private trajectories that are approximately indistinguishable. 
The definition of ``nearby" is given by an adjacency relation using the Hamming distance~\cite{Schulz2003distance} denoted by $d(w,v)$, which is a metric that measures the minimum number of substitutions that can be applied to a trajectory $v$ to convert it to $w$.

\begin{definition}[Adjacency]
\label{dfn:adjacency}
Fix a length $T\in\mathbb{N}^+$ and an adjacency parameter $\adjParam\in\mathbb{N}^+$. For an MDP with state space $\mdpStateSpace^i$, the adjacency relation on $(\mdpStateSpace^i)^T$ is $\adj_{T,\adjParam}=\{(v,w)\in (\mdpStateSpace^i)^T \times (\mdpStateSpace^i)^T\ |\ d(v,w)\leq \adjParam\}.$
\end{definition}

This adjacency relation specifies which trajectories are ``nearby'' and thus specifies pairs of trajectories that differential privacy must make approximately indistinguishable. Two $T-$length local trajectories in $(\mdpStateSpace^i)^T$ are adjacent if the Hamming distance between them is less than or equal to $k$. We next introduce the definition of word differential privacy, which guarantees that given a private trajectory, recipients are unlikely to distinguish between the underlying sensitive trajectory and other adjacent trajectories with high confidence. 

\begin{definition}[Word Differential Privacy~\cite{chen2022differential}]
\label{dfn:word_dp}
Fix a probability space $(\Omega,\mathcal{F},\mathbb{P}),$ an adjacency parameter $k\in\mathbb{N}^+,$ a length $T\in\mathbb{N}^+$ and a privacy parameter $\privacyLevel>0$. For an MDP with state space $\mdpStateSpace^i$, a mechanism $M:(\mdpStateSpace^i)^T\times\Omega\to \distribution((\mdpStateSpace^i)^T)$ is $\privacyLevel$-word differentially private if, for all trajectories $(v,w)\in \adj_{T,k}$ and all $L\subseteq (\mdpStateSpace^i)^T,$ it satisfies
$\mathbb{P}[M(v)\in L]\leq e^\privacyLevel \mathbb{P}[M(w)\in L].$
\end{definition}

The privacy parameter $\privacyLevel$ controls the strength of privacy and a smaller $\privacyLevel$ implies stronger privacy. In the literature, $\privacyLevel$ typically ranges from 0.01 to 10~\cite{Hsu2014DifferentialPA}. 

\section{Problem Formulation and Assumptions}
\label{sec:prob_form}
In this section, we state and analyze the problem of privatizing inter-agent communications in a cooperative Markov game. We begin with the problem statements.
Consider $N$ agents playing a cooperative Markov game with a reach-avoid objective as introduced in~\S\ref{sec:prelim_games}.

\begin{problem}
\label{prb:mechanism}
Design an online privacy mechanism that provides $\privacyLevel$-word differential privacy (Definition~\ref{dfn:word_dp}) for the state trajectory~$h_t^i=\mdpState_1^i \mdpState_2^i \dots \mdpState_t^i$ of agent $i$ in real time, i.e., without knowledge of $\mdpState^i_{t+1},\mdpState^i_{t+2},\dots$ at time $t.$ The mechanism should ensure that the private trajectory is still feasible with respect to the dynamics of the underlying MDP $\mdp^i$.
\end{problem}

\begin{problem}
    \label{prb:execution}
    Define an algorithm for the decentralized execution of local policies \(\lbrace \localPolicy^{i} \rbrace_{i=1}^{\numAgents}\) under private communications.
\end{problem}

\begin{problem}
\label{prb:bound}
Given a collection of local policies \(\lbrace \localPolicy^{i} \rbrace_{i=1}^{\numAgents}\), provide a bound on the probability of success under private communications $\gameValue^{\private}$. Use this bound to analyze the tradeoffs between privacy and performance in the multiagent system.
\end{problem}

\begin{problem}
\label{prb:synthesis}
Synthesize policies for the multiagent system that achieve high performance under strong levels of privacy, by taking into account the tradeoffs analyzed in Problem~\ref{prb:bound}.
\end{problem}

\paragraph{Privacy Assumptions:} We formalize what information agents provide to the central planner and what information they hide. We then illustrate this setting with an example.

We assume that each agent trusts a central planner to design local policies. Each agent allows the planner to access its individual MDP denoted as $\mdp^i$ for each $i\in[N]$. The planner also has knowledge of the game's objective, which can be specified as reach and avoid sets $\targetSet$ and $\deadSet$ or a reward function. The central planner provides each agent with a local policy $\localPolicy^{i}$, where ``local" refers to a mapping from the joint state space $\gameStateSpace$ to the local action space $\mdpActionSpace^i$. These local policies are assumed to be stationary and the action distribution of an agent is independent of the actions of the other agents given the joint state. This means that the central planner will not synthesize policies that compromise privacy in the sense that agent $i$ does not gain knowledge of any other agent's actions by sampling its own local policy $\localPolicy^i.$ In addition, we assume that the initial joint state, $\bm{\mdpInitialState}$, is public information.

We also assume that the agents do not fully trust each other. Knowing another agent's transition probabilities $\mdpTransition^i$, actions $a^i$, or rewards can harm that agent's privacy. However, each agent only needs the state information of the other agents to execute its local policy. Therefore, we assume that the agents do not have access to each other's transition probabilities or actions. The agents also do not observe whether the reach-avoid specification is satisfied or violated. Each agent \emph{only} receives the private state information from the rest of the network and a policy from the central planner. This prevents agents from being able to control the other agents' sensitive state trajectories. We also note that agents can know the state space and feasible state transitions of other agents, i.e., the support of $\mdpTransition^i(\mdpState^{i}, \mdpAction^{i})$, without compromising privacy. For example, two rideshare drivers know the possible locations of each other and how one another can transition through those locations, but this knowledge does not prevent the drivers from protecting their location information from each other.

Lastly, the methods presented in this paper can be applied when each agent has a different privacy level, i.e., different values of $\privacyLevel$. However, for convenience we assume that each agent has the same privacy parameter~$\privacyLevel.$ 

\begin{example}
In this example, the sensitive information is the location of two rideshare drivers, Alice and Bob. 
A central planner generates local policies for Alice and Bob to optimize its own objective.
To protect their privacy, Alice and Bob use differential privacy to communicate their locations to each other. For example, they can randomize their location data before sharing it so that their true locations are not revealed. Alice and Bob do not need to know each other's preferences or constraints, only the private state information that they communicate.

With this private information, Alice and Bob can then execute the local policies synthesized by the central planner. However, because they are sharing perturbed location data, the local policies may not be executed as efficiently as they could be if they had access to each other's true locations.

To mitigate this loss in performance, the central planner can use the methods developed in our paper to synthesize a collection of local policies that takes into account the effects of privacy on performance. 
This will allow them to balance the need for privacy with the need for efficient policy execution, and ensure that the passengers are picked up as quickly as possible while preserving the drivers' privacy.
\end{example}

\section{Implementing Local Policies with Private Communications}\label{sec:implemntation}
In this section, we solve Problems~\ref{prb:mechanism} and~\ref{prb:execution}. Specifically, in~\S\ref{subsec:privacy_implementation}, we modify the online mechanism for Markov chains from~\cite{chen2022differential} to privatize state trajectories of an MDP. Then, in~\S\ref{subsec:policy_implementation} we detail how each agent can use other agents' private state information to execute its local policy.

\subsection{Implementing Differential Privacy}\label{subsec:privacy_implementation}
We enforce privacy on a per-agent basis. That is, we develop a mechanism for agent $i$ to share its local state trajectory $\mdpStateSeq_t^i=\mdpState_1^i \mdpState_2^i \dots \mdpState_t^i\in (\mdpStateSpace^i)^t$ in real time while satisfying $\privacyLevel$-word differential privacy from Definition~\ref{dfn:priv_mech}. Here, ``real time" means that the private string will be generated symbol by symbol. To achieve this, agent $i$ will only share a private state trajectory $\Tilde{\mdpStateSeq}_t^i=\privateMdpState_1^i \privateMdpState_2^i \dots \privateMdpState_t^i\in (\mdpStateSpace^i)^t.$ To generate $\Tilde{\mdpStateSeq}_t^i$ in real time, agent $i$ uses an online mechanism $M_{\mdpStateSeq_t^i}$ to generate an individual private state $\privateMdpState_t^i$ at each time step $t.$ 

Each agent needs to communicate its private state with every other agent at every time step $t$ to allow agents to execute their policies. However, the differential privacy guarantee of Definition~\ref{dfn:word_dp} holds over the entire $T-$length state trajectory. This means that even though agents are communicating at each time step, we provide privacy to their entire $T-$length trajectories. We now define the online privacy mechanism that ensures the differential privacy over state trajectories.

\begin{definition}[Online Mechanism~\cite{chen2022differential}]
\label{dfn:priv_mech}
Fix a probability space $(\Omega,\mathcal{F},\mathbb{P})$ and an MDP $\mdp^i=(\mdpStateSpace^i,\mdpInitialState^i,\mdpActionSpace^i,\mdpTransition^i)$. Given a state trajectory $\mdpStateSeq_t^i=\mdpState_
1^i \mdpState_2^i \dots \mdpState_t^i\in (\mdpStateSpace^i)^t,$ with an initial state $\mdpState_I^i$, define the online mechanism $M_{\mdpStateSeq_t^i}$ that generates a private trajectory $\Tilde{\mdpStateSeq}^i_t=\privateMdpState_1^i \privateMdpState_2^i \dots \privateMdpState_t^i\in (\mdpStateSpace^i)^t$ such that $\privateMdpState_t^i$ is sampled from the distribution $\mathbb{P}[\privateMdpState_t^i]=\probabilityMeasure^i_\privacyLevel(\privateMdpState_t^i|\mdpState_t^i,\privateMdpState_{t-1}^i)$ where $\probabilityMeasure^i_\privacyLevel$ is computed by Algorithm~\ref{alg:privacy_construction}.
\end{definition}
\begin{algorithm}
\caption{Online Mechanism Construction}
\label{alg:privacy_construction}
\KwIn{Probability of true transition $\trueStateProb_\privacyLevel$}
\KwOut{$\probabilityMeasure^i_{\privacyLevel}$}
\For{$(\mdpState^{i}_t, \privateMdpState^{i}_{t-1}, \privateMdpState^{i}_t)\in \mdpStateSpace^i\times \mdpStateSpace^i \times \mdpStateSpace^i$}{
            \uIf{$\mdpState^{i}_t=\privateMdpState^{i}_t$\ and $\beta(\privateMdpState^{i}_t,\privateMdpState^{i}_{t-1})=1$}{
                $\probabilityMeasure^i_{\privacyLevel}(\privateMdpState^{i}_t\ |\ \mdpState^{i}_t,\privateMdpState^{i}_{t-1})=\trueStateProb_\privacyLevel(\privateMdpState^{i}_{t-1}).$
            }
            \uElseIf{$\mdpState^{i}_t\neq \privateMdpState^{i}_t$\ and $\beta(\privateMdpState^{i}_t,\privateMdpState^{i}_{t-1})=1$}{
                $\probabilityMeasure^i_{\privacyLevel}(\privateMdpState^{i}_t\ |\ \mdpState^{i}_t,\privateMdpState^{i}_{t-1})=\frac{1-\trueStateProb_\privacyLevel(\privateMdpState^{i}_{t-1})\beta(\mdpState^{i}_t,\privateMdpState^{i}_{t-1})}{\rho(\privateMdpState^{i}_{t-1})-\beta(\mdpState^{i}_t,\privateMdpState^{i}_{t-1})}.$
            }
            \Else{
            $\probabilityMeasure^i_{\privacyLevel}(\privateMdpState^{i}_t\ |\ \mdpState^{i}_t,\privateMdpState^{i}_{t-1})=0.$}
}
\end{algorithm}
In Algorithm~\ref{alg:privacy_construction}, the feasibility indicator function $\beta$ is defined for all $\mdpState^i,\mdpStateAlt^i\in\mdpStateSpace^i$ as
\begin{equation*}
    \beta(\mdpState^i,\mdpStateAlt^i)=\begin{cases}
    1,\ \text{if}\ \exists \mdpAction^i\in\mdpActionSpace^i\ s.t.\ \mdpTransition^i(\mdpStateAlt^i,\mdpAction^i,\mdpState^i)>0,\\
    0,\ otherwise,
    \end{cases}\\
\end{equation*}
and the out-degree $\rho$ is defined for each state $\mdpState^i\in\mdpStateSpace$ as
$\rho(\mdpState^i) = \lvert\{\mdpStateAlt^{i}\in\mdpStateSpace^i\ |\  \exists \mdpAction^i\in\mdpActionSpace^i\ s.t.\ \mdpTransition^i(\mdpState^i,\mdpAction^i,\mdpStateAlt^{i})>0\} \rvert.$

Definition~\ref{dfn:priv_mech} and Algorithm~\ref{alg:privacy_construction} define a privacy mechanism in the form of a conditional probability distribution $\probabilityMeasure^i_{\privacyLevel}.$ To implement the mechanism agent $i$ samples a private output $\privateMdpState^i_t$ from the probability distribution $\probabilityMeasure^i_{\privacyLevel}(\cdot|\ \mdpState_t^i,\privateMdpState_{t-1}^i)$ at each time step $t.$ The mechanism is constructed such that the probability $\probabilityMeasure^i_{\privacyLevel}(\privateMdpState_{t}^i\ |\ \mdpState_t^i,\privateMdpState_{t-1}^i)$ is positive if $\privateMdpState^i_t$ is feasible from the most recent private state $\privateMdpState^i_{t-1},$ and $0$ otherwise. This prevents the mechanism from outputting private trajectories that are not feasible with respect to the dynamics of $\mdp^i.$

When the true, sensitive state $\mdpState_t^i$ is feasible from the previous private output $\privateMdpState^i_{t-1},$ the mechanism outputs $\mdpState_t^i$ with probability $\trueStateProb_\privacyLevel(\privateMdpState^i_{t-1})$ and outputs any other feasible state with a uniform probability whose sum is equal to $1-\trueStateProb_\privacyLevel(\privateMdpState^i_{t-1})$. We refer to the event of outputting the sensitive state $\mdpState_t^i$ at time $t$ as a ``true transition" and $\trueStateProb_\privacyLevel(\privateMdpState^i_{t-1})$ as the ``probability of true transition". In~\S\ref{sec:analytical_results}, we establish a requirement for this mechanism to achieve $\privacyLevel$-word differential privacy.

\subsection{Private Policy Execution}\label{subsec:policy_implementation}
In this section, we solve Problem~\ref{prb:execution} and define an algorithm for the decentralized execution of local policies \(\lbrace \localPolicy^{i} \rbrace_{i=1}^{\numAgents}\) under private communications (Algorithm \ref{alg:policy_exec}). 

\begin{algorithm}
\caption{Policy Execution with Private Communications}
\label{alg:policy_exec}
\SetKwBlock{DoParallel}{Every agent \(i\) does in parallel}{end}
\textbf{Input for every agent \(i\):} Local policy \(\localPolicy^{i}\)

Set \(\privateMdpState^{i}_{0} = \mdpInitialState^{i}\) for all \(i\in [\numAgents]\).

\For{$t=0,1,\dots$}{

\DoParallel{
    Set \(\hat{\gameState}_{t,i} = (\Tilde{\mdpState}^{1}_{t}, \ldots, \Tilde{\mdpState}^{i-1}_{t}, \mdpState^{i}_{t}, \Tilde{\mdpState}^{i+1}_{t}, \ldots, \Tilde{\mdpState}^{\numAgents}_{t})\).
       
    Sample an action $\mdpAction_{t}^{i} \sim \localPolicy^{i}(\hat{\mdpState}_{t,i}).$
    
    Execute $\mdpAction_{t}^i$ and transition to $\mdpState_{t+1}^i\sim\mdpTransition^i(\mdpState_{t}^i,\mdpAction_{t}^i).$

    Share $\privateMdpState_{t+1}^{i}\sim\probabilityMeasure^i_{\privacyLevel}(\cdot|\mdpState_{t+1}^i,\privateMdpState_{t}^i)$ with other agents. 
}

}
\end{algorithm}
Since the agents are communicating potentially false information, no agent truly knows the true joint state of the network. Thus, the network of agents cannot execute the local policies exactly when communications are private. To overcome this, each agent maintains an estimate of the joint state and makes its own action decisions based on this estimate. In this work, we assume that each agent takes the private information as the truth, i.e., each agent estimates the joint state as the private information it receives. In detail, agent \(i\) knows its own local state \(\mdpState^{i}_{t}\) and the private state \(\privateMdpState^{j}_{t}\) of every other agent \(j\) at time \(t\). Agent $i$'s estimate of the joint state is denoted by \(\hat{\mdpState}_{t,i} = (\Tilde{\mdpState}^{1}_{t}, \ldots, \Tilde{\mdpState}^{i-1}_{t}, \mdpState^{i}_{t}, \Tilde{\mdpState}^{i+1}_{t}, \ldots, \Tilde{\mdpState}^{\numAgents}_{t})\). Since agent \(i\) does not know the true joint state \((\mdpState^{1}_{t}, \ldots, \mdpState^{\numAgents}_{t})\), agent \(i\) samples an action \(\mdpAction_{t}^{i}\) for itself from \(\localPolicy^{i}\) using its state estimate \(\hat{\mdpState}_{t,i}\). We note that the agents do not communicate during the action selection phase since the local policies are independent given the joint state. After choosing an action \(\mdpAction_{t}^{i}\), agent \(i\) executes this action and transitions to a next state \(\mdpState_{t+1}^i\). In the next time step \(t+1\), agent \(i\) samples a private state \(\privateMdpState^{i}_{t}\) using~$\probabilityMeasure^i_{\privacyLevel}$ and shares this private state with the other agents. Then, the agents again sample and execute their local actions.

\section{Privacy and Performance Tradeoffs}
\label{sec:analytical_results}
We address Problem \ref{prb:bound} in this section and analyze the tradeoff between performance and privacy when executing a collection of local policies with private communications. 

We provide the following lemma from \cite{chen2022differential} which establishes $\privacyLevel$-word differential privacy of the agents' state trajectories generated by the Online mechanism. 

\begin{lemma}[\cite{chen2022differential}]
\label{lem:privacy_req}
Fix a length $T\in\mathbb{N}^+$, an adjacency parameter $\adjParam\in\mathbb{N}^+,$ and a privacy parameter $\privacyLevel\geq0.$ The online mechanism appearing in Definition~\ref{dfn:priv_mech} is $\privacyLevel$-word differentially private (Definition~\ref{dfn:word_dp}) with respect to the Adjacency relation $\adj_{T,\adjParam}$ in Definition~\ref{dfn:adjacency} if $\trueStateProb_\privacyLevel(\privateMdpState^i_{t-1})$ satisfies
\begin{equation}
    \trueStateProb_\privacyLevel(\privateMdpState^i_{t-1})=1/({(\rho(\privateMdpState^i_{t-1})-1)e^{-\nicefrac{\privacyLevel}{\adjParam}+1})}.
\end{equation}
\end{lemma}

Having established the differential privacy guarantees of Algorithm \ref{alg:policy_exec}, we now focus on performance guarantees. In order to succeed under private communications, the agents' local policies should be as indifferent as possible to the other agents' states. In other words, agents' behaviors should be made nearly independent from each other. 

The collection of local policies induce a joint policy \(\jointPolicy = \lbrace \localPolicy^{i} \rbrace_{i=1}^{\numAgents}\). To measure the dependencies between the agents, we use a quantity called the ``total correlation" of the joint policy~\cite{Karabag2022}. Let $\gameStateRandomVar_t$ be a random variable denoting the joint state of the agents at time $t$ under the joint policy \(\jointPolicy\) with no privatization, $\gameActionRandomVar_t$ be a random variable denoting the joint action of the agents at time $t$, $\mdpStateRandomVar_t^i$ be a random variable denoting the state of Agent $i$ at time $t$, and let $\mdpActionRandomVar_t^i$ be a random variable denoting the action of Agent $i$ at time $t$.
The total correlation $\totalCorrelation_{\jointPolicy}$ of a joint policy $\jointPolicy = \lbrace \localPolicy^{i} \rbrace_{i=1}^{\numAgents}$ is
\begin{equation}
\totalCorrelation_{\jointPolicy} = \Sigma_{i=1}^N \entropy(\mdpStateRandomVar_0^i\mdpActionRandomVar_0^i\ldots\mdpStateRandomVar_\randomReachTime^i)-\entropy(\gameStateRandomVar_0\gameActionRandomVar_0\ldots\gameStateRandomVar_\randomReachTime)\label{eq:proof_tc_def}
\end{equation}
where \(\randomReachTime\) denotes the random hitting time to \(\targetSet \cup \deadSetPrime\), i.e., the effective end of the trajectory in terms of the reach-avoid specification~\cite{Karabag2022}.

We have the following result that relates the success probability under private communications to the success probability under truthful communications (i.e., no privacy). 
The proof of the theorem is included in the supplementary material.

\begin{theorem}\label{thm:performance_bound}
Fix a privacy parameter $\epsilon>0$ and adjacency parameter $\adjParam.$
Given $\numAgents$ agents implementing a collection of local policies $\jointPolicy = \lbrace \localPolicy^{i} \rbrace_{i=1}^{\numAgents}$ with private communications according to Algorithm~\ref{alg:policy_exec}, let $\gameValue^{\private}$ be the success probability under private communications and let $\gameValue^{\oblivious}$ be the success probability under truthful communications, i.e.,  no privacy. Then,
\begin{align}
\label{eq:thm_1_eq}
\gameValue^{\private}& \geq \gameValue^{\oblivious}
    -\sqrt{1 - e^{-\totalCorrelation_{\jointPolicy}} \left((\outdegree_{m}-1)e^{-{\epsilon}/{\adjParam}}+1\right)^{\numAgents \expectedLength^{\oblivious}}},
\end{align}
where $\totalCorrelation_{\jointPolicy}$ is defined in~\eqref{eq:proof_tc_def}, $\outdegree_{m}=\max_{i \in [\numAgents], \mdpState^{i}\in \mdpStateSpace^{i}}\outdegree(\mdpState^{i})$ is the max out-degree, \(\expectedLength^{\oblivious} = \expectation_{\gamePath \sim \gamePathDist^{\oblivious}}[\len(\gamePath)]\) is the expected joint trajectory length when $\jointPolicy$ is executed with no privacy, and $\Gamma^{\oblivious}$ is the probability distribution over joint trajectories induced by the joint policy executed with no privacy.
\end{theorem}
The term \(\left((\outdegree_{m}-1)e^{-{\epsilon}/{\adjParam}}+1\right)^{\numAgents \expectedLength^{\oblivious}}\) in Theorem \ref{thm:performance_bound} represents the probability of events where the private state trajectories are the same as the true state trajectories.  
The term \(e^{-\totalCorrelation_{\jointPolicy}}\) in Theorem \ref{thm:performance_bound} is a proxy to account for the events where the private state trajectories are different from the true state trajectories. In these events, the agents can still succeed if the local policies are independent of the other agents' states. A lower total correlation implies lower dependencies between the agents, and that the agents are more likely to succeed. We note that the equality holds in~\eqref{eq:thm_1_eq} when agents communicate truthfully, i.e., $\privacyLevel=\infty$, and each agent acts totally independently from other agents, i.e., $\totalCorrelation_{\jointPolicy} =0$. 

\section{Policy Synthesis}
\label{sec:policy_synthesis}
In preceding sections, we discussed the execution of a fixed collection of local policies and analyzed the performance of these policies under private communications. We  now present the synthesis of a collection of local policies $\jointPolicy = \lbrace \localPolicy^{i} \rbrace_{i=1}^{\numAgents}$ that remains performant under private communications.

We aim to maximize the reach-avoid probability under private communications by minimizing the lower bound on $\gameValue^{\private}$ given in Theorem~\ref{thm:performance_bound}. Since the bound is complex in nature and it is a monotone function of its variables, we instead solve the following optimization problem: 
\begin{equation}
    \text{sup}_{\jointPolicy} \gameValue^{\oblivious} - \delta\expectedLength^{\oblivious}-\beta\totalCorrelation_{\jointPolicy}\label{eq:proxy_optimization_problem}
\end{equation}
where $\delta>0$ and $\beta>0$ are constants. 

Using the stationarity of \(\jointPolicy\), the optimization problem given in \eqref{eq:proxy_optimization_problem} can be represented with occupancy measure variables \(\occupancyVar_{\gameState, \gameAction}\) of the joint state-action space~\cite{Karabag2022}. We refer interested readers to \cite{Karabag2022} for the details of this optimization problem. The objective function of~\eqref{eq:proxy_optimization_problem} contains convex and concave functions of the occupancy measure variables that can be solved using the concave-convex procedure~\cite{Lanckriet2009convergence,Yuille2001procedure} for a local optimum.

After solving for the optimal $\occupancyVar_{\gameState,\gameAction}^*$ of the occupancy measure variables, we compute the local policies. Recall that we assumed in~\S\ref{sec:prob_form}, the agents are given local policies \(\localPolicy^{i}\) that have independent action distributions given the joint state. In order to compute local policies \(\localPolicy^{i}(\gameState, \mdpAction^{i})\), we marginalize the joint occupancy measure. Formally, we have \[\localPolicy^i(\gameState, \mdpAction^{i}) = \bigl( \Sigma_{\gameAction^{-1} \in \gameActionSpace^{-1}} \ \occupancyVar^{*}_{\gameState, (\mdpAction^{i}, \gameAction^{-i})}\bigr) / \bigl(\Sigma_{\gameAction \in \gameActionSpace} \ \occupancyVar^{*}_{\gameState, \gameAction}\bigr).\]

We note that we can alternatively enforce the independence of local policies (given the joint state) during synthesis procedure instead of postprocessing joint occupancy variables. The details of this procedure are given in the supplementary material.

\section{Numerical Experiments}
\label{sec:expirements}

Numerical experiments demonstrate the robustness to private communication enjoyed by the policies synthesized using the procedure described in \S \ref{sec:policy_synthesis}.
In each experiment, we solve \eqref{eq:proxy_optimization_problem} to synthesize minimum-dependency local policies \(\{\localPolicy_{MD}^{i}\}_{i=1}^{\numAgents}\) for the agents in the team.
We use \(\jointPolicy_{MD}\) to denote the joint policy that results from the concurrent execution of these local policies, described in \S \ref{subsec:policy_implementation}.

We compare the performance of \(\jointPolicy_{MD}\) to that of a collection of baseline local policies \(\{\localPolicy_{base}^{i}\}_{i=1}^{\numAgents}\), which are synthesized by optimizing the team's performance without taking the total correlation value into account.
That is, by solving \eqref{eq:proxy_optimization_problem} with \(\delta\) and \(\beta\) set to zero.
We use \(\jointPolicy_{base}\) to refer to the joint policy resulting from the concurrent execution of \(\{\localPolicy_{base}^{i}\}_{i=1}^{\numAgents}\).

Code to reproduce all experiments and analysis is available at \href{https://github.com/cyrusneary/differential_privacy_in_mas}{https://github.com/cyrusneary/differential\_privacy\_in\_mas}.

\subsection{Two-Agent Navigation Example}
\label{sec:experiments_navigation}

We begin by considering the multiagent navigation example introduced in \cite{Karabag2022}.
Two agents operate in a common environment, and each must navigate to a target location while avoiding collisions with its teammate.
To reach their target locations, the agents must navigate past each other by passing through one of two narrow corridors.
The act of jointly navigating the corridors without colliding necessitates coordination between the agents.

The environment is implemented as a grid of cells, each of which corresponds to an individual local state.
At any given timestep, each agent takes one of five separate actions: move left, move right, move up, move down, or remain in place.
Each agent slips with probability \(0.05\) every time it takes an action, resulting in the agent moving to one of its neighboring states instead of its intended target state.

While synthesizing \(\jointPolicy_{MD}\), we set the values of the coefficients \(\expectedLengthCoef\) and \(\totalCorrelationCoef\) in \eqref{eq:proxy_optimization_problem} to \(0.01\) and \(0.4\) respectively. 
These values were selected to strike a balance between the optimization objective's three competing terms.
We fix an adjacency parameter of \(\adjParam = 3\) while constructing the differential privacy mechanisms used in all experiments.

\begin{figure}[t!]
    \centering
    % This file was created with tikzplotlib v0.9.12.
\begin{tikzpicture}

\begin{axis}[
width=0.95*\columnwidth, 
height=4.0cm,
legend cell align={left},
legend columns = 1,
legend style={
  fill opacity=0.8,
  draw opacity=1,
  text opacity=1,
  at={(3.05cm, 0.2cm)},
  anchor=south west,
  draw=white!80!black
},
tick align=inside,
tick pos=left,
x grid style={white!69.0196078431373!black},
xlabel={\footnotesize Number of Convex-Concave Iterations},
xmajorgrids,
xmin=0.0, xmax=80.0,
xtick style={color=black},
y grid style={white!69.0196078431373!black},
ylabel={\footnotesize Success Prob.},
ymajorgrids,
ymin=0.0, ymax=1.0562822829432,
ytick style={color=black},
ytick={0.0, 0.2, 0.3, 0.4, 0.5, 0.6, 0.7, 0.8, 0.9, 1.0},
yticklabels={0.0, 0.2, ,0.4, ,0.6, , 0.8, , 1.0},
]

\input{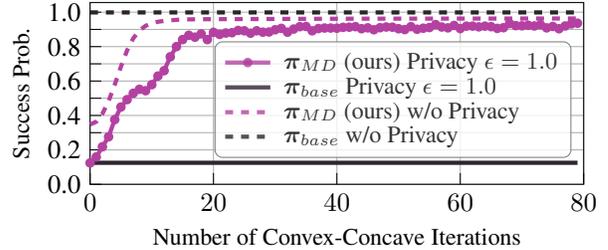}

\end{axis}

\begin{customlegend}[
    legend columns=1, 
    legend cell align={left},
    legend style={
        align=left, 
        column sep=0.5ex, 
        font=\footnotesize, 
        draw=white!50!black,
        fill=white,
        fill opacity=0.8,
        rounded corners=1mm,
        at={(64.0mm, 19.0mm)},
        row sep=-0.08cm,
    }, 
    legend entries={\(\jointPolicy_{MD}\) (ours) Privacy \(\privacyLevel = 1.0\), \(\jointPolicy_{base}\) Privacy \(\privacyLevel = 1.0\), \(\jointPolicy_{MD}\) (ours) w/o Privacy, \(\jointPolicy_{base}\) w/o Privacy,}
]
\addlegendimage{ultra thick, minDependencyPolicy, mark=*, mark size=1, mark options={solid}}
\addlegendimage{ultra thick, baselinePolicy}
\addlegendimage{ultra thick, minDependencyPolicy, dashed}
\addlegendimage{ultra thick, baselinePolicy, dashed}
\end{customlegend}

\end{tikzpicture}
    \caption{
    Probability of task success as a function of the number of iterations of the policy synthesis procedure for the two-agent navigation experiment. 
    In contrast to the baseline policy, \(\jointPolicy_{MD}\) achieves a high probability of success, even under private communications.
    }
    \label{fig:multiagent_navigation_results_during_policy_synthesis}
\end{figure}

\begin{figure}[t!]
    \centering
    % This file was created with tikzplotlib v0.10.1.
\begin{tikzpicture}
\definecolor{darkgray176}{RGB}{176,176,176}

\begin{axis}[
width=0.95*\columnwidth, 
height=3.5cm,
tick align=inside,
tick pos=left,
x grid style={darkgray176},
xmajorgrids,
xmin=2.3184187339362, xmax=5.28759802936976,
xtick style={color=black},
y grid style={darkgray176},
ymajorgrids,
ymin=0.0, ymax=1.0,
ytick style={color=black},
ytick={0.0, 0.1, 0.2, 0.3, 0.4, 0.5, 0.6, 0.7, 0.8, 0.9, 1.0},
yticklabels={0.0, , 0.2, ,0.4, ,0.6, , 0.8, , 1.0},
ylabel style={align=center},
ylabel={\footnotesize Success Prob.},
xlabel={\footnotesize Total Correlation Value of the Synthesized Policy \(\policy_{MD}\)},
]

\addplot [ultra thick, minDependencyPolicy, mark=*, mark size=2, mark options={solid}]
table {%
5.15263533412278 0.124
4.68315799524924 0.159
4.39170326856943 0.218
4.21395294917153 0.276
4.11560418830195 0.376
4.0664990574179 0.449
4.0382724880815 0.492
4.01201964692952 0.529
3.98033289098386 0.552
3.94126881661411 0.544
3.89059492078465 0.579
3.81548532419666 0.629
3.69152974862599 0.66
3.49555978179105 0.741
3.25220033928505 0.8
3.04760075299328 0.841
2.9266338974042 0.866
2.86263007993679 0.855
2.82518849700872 0.874
2.79879143253659 0.841
2.7773051196462 0.883
2.75854818176313 0.871
2.74171583915934 0.879
2.72649026216253 0.881
2.7125888366637 0.894
2.69977822678362 0.892
2.6879019550388 0.878
2.67688768770299 0.892
2.6666635470067 0.884
2.6571089641803 0.885
2.6480963282199 0.898
2.63950801507142 0.906
2.63123083634076 0.884
2.62315268700401 0.902
2.61516112642168 0.916
2.60714490897642 0.917
2.59899801058677 0.892
2.59062555345212 0.905
2.58195419778296 0.908
2.57294340178707 0.908
2.56360160687682 0.907
2.55399969661561 0.912
2.54427584498694 0.913
2.53463171849739 0.916
2.5253067388059 0.904
2.5165405771671 0.909
2.50853021230751 0.898
2.50139860047124 0.91
2.49518314060074 0.914
2.48984689064893 0.913
2.48530124554364 0.904
2.48142608181387 0.899
2.47808862349647 0.916
2.47519861043089 0.901
2.47274751258593 0.907
2.47066588715846 0.911
2.46886363478232 0.928
2.46728312502176 0.914
2.46588533392269 0.909
2.46464129099376 0.923
2.46352823148981 0.905
2.4625290607641 0.919
2.46162816549776 0.919
2.4608139556685 0.915
2.46007532058983 0.919
2.4594025213932 0.916
2.45878810585337 0.906
2.45822372805964 0.91
2.45770269694576 0.934
2.45721876173375 0.922
2.45676623525647 0.921
2.45634058879823 0.924
2.45593538427419 0.92
2.45554848171814 0.927
2.45517465206821 0.909
2.45481029851174 0.917
2.45445227695792 0.904
2.45409684651753 0.92
2.45374081260005 0.932
2.45338142918318 0.936
};
\end{axis}

\node at (4.6cm, 1.6cm) [font=\footnotesize, align=left] {Privacy parameter \(\privacyLevel = 1.0\)};

\end{tikzpicture}
    % \vspace{-0.5cm}
    \caption{
    Probability of team success under private communications as a function of the total correlation of the synthesized policies.
    % \(\policy_{MD}\).
    }
    \label{fig:success_prob_vs_corr_two_agent_navigation}
\end{figure}
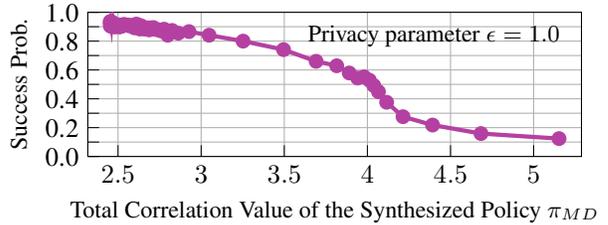

\begin{figure}[t!]
    \centering
    % This file was created with tikzplotlib v0.10.1.
\begin{tikzpicture}

\definecolor{darkgray176}{RGB}{176,176,176}

\begin{axis}[
width=0.95*\columnwidth, 
height=3.5cm,
tick align=inside,
tick pos=left,
x grid style={darkgray176},
xmajorgrids,
xmin=-0.0, xmax=10.0,
xtick style={color=black},
y grid style={darkgray176},
ymajorgrids,
ymin=0.0, ymax=1.0,
ytick style={color=black},
ytick={0.0, 0.1, 0.2, 0.3, 0.4, 0.5, 0.6, 0.7, 0.8, 0.9, 1.0},
yticklabels={0.0, , 0.2, ,0.4, ,0.6, , 0.8, , 1.0},
ylabel={\footnotesize Success Prob.},
xlabel={\footnotesize Privacy Parameter \(\privacyLevel\)},
]
\addplot [ultra thick, minDependencyPolicy, mark=diamond*, mark size=3, mark options={solid}]
table {%
0.01 0.911
1.009 0.907
2.008 0.916
3.007 0.924
4.006 0.923
5.005 0.929
6.004 0.945
7.003 0.934
8.002 0.925
9.001 0.939
10 0.958
};
\addplot [ultra thick, baselinePolicy, mark=diamond*, mark size=3, mark options={solid}]
table {%
0.01 0.102
1.009 0.141
2.008 0.142
3.007 0.172
4.006 0.265
5.005 0.346
6.004 0.469
7.003 0.588
8.002 0.689
9.001 0.792
10 0.852
};
\end{axis}

\begin{customlegend}[
    legend columns=1, 
    legend style={
        align=left, 
        column sep=1ex, 
        font=\footnotesize, 
        draw=white!50!black,
        fill=white,
        fill opacity=0.8,
        rounded corners=1mm,
        at={(23.0mm, 15.0mm)},
    }, 
    legend entries={\(\jointPolicy_{MD}\), \(\jointPolicy_{base}\)}
]
\addlegendimage{ultra thick, minDependencyPolicy, mark=diamond*, mark size=3, mark options={solid}}
\addlegendimage{ultra thick, baselinePolicy, mark=diamond*, mark size=3, mark options={solid}}
\end{customlegend}

\end{tikzpicture}
    % \vspace{-0.1cm}
    \caption{Probability of team success under a variety of levels of privacy. Smaller values of the privacy parameter \(\privacyLevel\) correspond to a stronger level of privacy.}
    \label{fig:success_prob_vs_privacy_level_two_agent_navigation}
\end{figure}
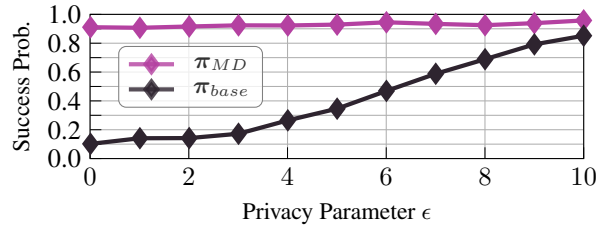

\paragraph{Minimum-dependency policies are \(80\%\) more performant than the baseline under private communications.}
Figure \ref{fig:multiagent_navigation_results_during_policy_synthesis} illustrates the probability of success of \(\jointPolicy_{MD}\) throughout policy synthesis. 
We plot the success probability resulting from both private \((\privacyLevel = 1.0)\) and non-private (the agents do not apply the privacy mechanism to their state trajectories) implementations of communication.
For comparison, we also plot the results of \(\jointPolicy_{base}\).
We estimate the plotted probability values by simulating \(1,000\) rollouts of the policies at each iteration, and computing the empirical rate at which the team reaches its target set.
While the baseline policy achieves a success probability of \(0.99\) under non-private communication, its success probability drops to \(0.13\) when communications are private.
By contrast, even under private communication, \(\jointPolicy_{MD}\) enjoys a probability of success of \(0.92\).

\paragraph{Lower total correlation values result in higher success probabilities under private communications.}

Figure \ref{fig:success_prob_vs_corr_two_agent_navigation} illustrates the team's success probability and the total correlation of each of the joint policies obtained throughout policy synthesis.
We observe that as the total correlation of \(\jointPolicy_{MD}\) decreases during policy synthesis, the policy's performance under private communications significantly increases.
This result provides a strong empirical justification for the use of total correlation as a regularizer during policy synthesis.

\paragraph{The performance of \(\jointPolicy_{MD}\) is robust to level of privacy enforced by the differential privacy mechanism.}
Recall that the parameter \(\privacyLevel\) controls the strength of privacy enforced by the differential privacy mechanism.
Lower values of \(\privacyLevel\) correspond to stronger levels of privacy---the mechanism is more likely to perturb the state trajectories of the agents.
In Figure \ref{fig:success_prob_vs_privacy_level_two_agent_navigation} we observe that the performance \(\jointPolicy_{MD}\) remains consistently high, regardless of the value of \(\privacyLevel\).
By contrast, the performance of \(\jointPolicy_{base}\) is highly sensitive to \(\privacyLevel\); it decreases significantly for moderate to strong levels of privacy.

%%%%%%%%%%%%%%%%%%%%%%%%%%%%%%%%%%%%%%%%%%%%%%%%%%%%%%%%%%%%%%%
\subsection{Four-Agent SysAdmin Example}
\label{sec:experiments_sysadmin}

\begin{figure}[t!]
    \centering
    \tikzstyle{branch}=[fill,shape=circle,minimum size=5pt,inner sep=0pt]
\def\horizontaldistance{1.5cm}
\def\verticaldistance{2.0cm}
\def\nodedistance{2.5cm}
\def\branchdist{0.8cm}

\def\stateOutlineThickness{0.2mm}
\def\edgeThickness{0.2mm}

\tikzset{auto, ->, >=stealth', auto, node distance=\nodedistance, node/.style={scale=0.8, minimum size=0pt, inner sep=0pt}}
\tikzset{every loop/.style={min distance=8mm,in=45,out=135,looseness=10}}

% \resizebox{1.0\textwidth}{!}{
\begin{tikzpicture}[scale=1.0]

    % Draw the states
    \node[state, line width=\stateOutlineThickness, minimum size=0.7cm] (s_repair) {\(0\)};
    
    \node[state, line width=\stateOutlineThickness, right=\horizontaldistance of s_repair, font=\footnotesize, minimum size=0.7cm] (s_nominal) {\(1\)};
    
    \node[state, line width=\stateOutlineThickness, right=\horizontaldistance of s_nominal, align=center, font=\footnotesize, minimum size=0.7cm] (s_unhealthy) {\(2\)};
    
    \node[state, line width=\stateOutlineThickness, right=\horizontaldistance of s_unhealthy, align=center, font=\footnotesize, minimum size=0.7cm] (s_offline) {\(3\)};

    % Draw the transitions
    \path [draw, -latex, line width=\edgeThickness] (s_repair.east) -- node [above, yshift=0.0mm, xshift=0.0mm, align=left, font=\footnotesize] {wait, \(\probRepair\)} (s_nominal.west);
    
    \path [draw, -latex, line width=\edgeThickness] (s_nominal.east) -- node [above, yshift=0.0mm, xshift=0.0mm, align=left, font=\footnotesize] {wait, \(\probONB\)} (s_unhealthy.west);
    
    \path [draw, -latex, line width=\edgeThickness] (s_unhealthy.east) -- node [above, yshift=0.0mm, xshift=0.0mm, align=left, font=\footnotesize] {wait, \(\probOffline\)} (s_offline.west);
    
    % Paths to the repair state
    \path [draw, line width=\edgeThickness, -latex] (s_nominal.south) to [bend left, in=150, out=30] node [above, align=left, font=\footnotesize] {repair, \(1\)} (s_repair.south);
    
    \path [draw, line width=\edgeThickness] (s_unhealthy.south) to [bend left, in=150, out=25] node [above, align=left, font=\footnotesize, xshift=0.6cm, yshift=0.1cm] {repair, \(1\)} (s_repair.south);
    
    \path [draw, line width=\edgeThickness] (s_offline.south) to [bend left, in=150, out=20] node [above, align=left, font=\footnotesize, xshift=1.5cm, yshift=0.2cm] {repair, \(1\)} (s_repair.south);
    
    % Self loops
    \path (s_repair.north) edge [loop above] node [above, align=left, font=\footnotesize] {wait, \(1-\probRepair\)} (s_repair.north);
    \path (s_nominal.north) edge [loop above] node [above, align=left, font=\footnotesize] {wait, \(1-\probONB\)} (s_repair.north);
    \path (s_unhealthy.north) edge [loop above] node [above, align=left, font=\footnotesize] {wait, \(1-\probOffline\)} (s_repair.north);
    \path (s_offline.north) edge [loop above] node [above, align=left, font=\footnotesize] {wait, \(1\)} (s_repair.north);

\end{tikzpicture}
% }
    \vspace{-0.6cm}
    \caption{
    Local transition dynamics of the SysAdmin example. A label \((a,p)\) refers to a transition happening w.p. \(p\) under action \(a\).
    }
    \label{fig:sys_admin_local_dynamics}
\end{figure}
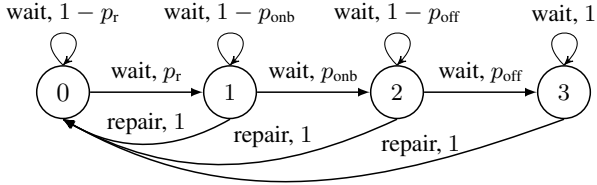

We now consider a variant of the multiagent system administration example from \cite{guestrin2003efficient,choudhury2021scalable}.
A collection of servers must coordinate to provide a consistent level of service, while simultaneously performing necessary maintenance.
Each server is modeled as an individual agent with four local states: nominal \(\mdpState^{i} = 1\), in need of repairs \(\mdpState^{i} = 2\), in repair \(\mdpState^{i} = 0\), and offline \(\mdpState^{i} = 3\).
At any timestep, each agent may choose to continue operation, or to initiate a repair.
We assume the local transition dynamics of the agents, illustrated in Figure \ref{fig:sys_admin_local_dynamics}, to be independent.

The team's task is to reach a target joint state in which all of the servers are operating nominally.
However, we impose the additional constraints that, at any given time during operation, the team is allowed at most two offline servers and at most two servers in the repair state.
If either of these constraints are violated, the team fails the task.

In this example, we set \(\probRepair = 0.9\), \(\probONB = 0.1\), and \(\probOffline = 0.1\), we set the values of the policy synthesis  coefficients \(\expectedLengthCoef\) and \(\totalCorrelationCoef\) to \(0.001\) and \(0.1\) respectively, and we use an adjacency parameter of \(\adjParam = 1\) in the differential privacy mechanism.

\paragraph{\(\jointPolicy_{MD}\) consistently outperforms \(\jointPolicy_{base}\) under a variety of initial system configurations and privacy levels.}

Figure \ref{fig:sys_admin_bar_chart} compares the probability of success achieved by the proposed minimum-dependency policy \(\jointPolicy_{MD}\), to that achieved by the baseline \(\jointPolicy_{base}\).
We test the team's performance under a number of different levels of privacy, and from a variety of initial system configurations. We observe that both \(\jointPolicy_{MD}\) and \(\jointPolicy_{base}\) achieve near-perfect performance under truthful communication (solid bars).
However, when communication is private, \(\jointPolicy_{MD}\) consistently outperforms \(\jointPolicy_{base}\).
In the considered initial configurations, even under the strongest level of privacy, \(\jointPolicy_{MD}\) achieves a probability of success of above \(87\) percent.

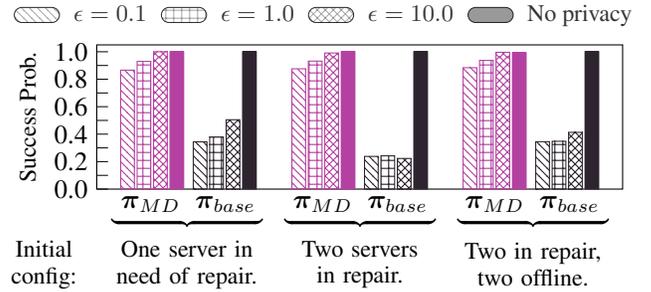
\begin{figure}[t!]
    \centering
    % This file was created with tikzplotlib v0.10.1.
\begin{tikzpicture}

\definecolor{crimson2143940}{RGB}{214,39,40}
\definecolor{darkgray176}{RGB}{176,176,176}
\definecolor{darkorange25512714}{RGB}{255,127,14}
\definecolor{forestgreen4416044}{RGB}{44,160,44}
\definecolor{gray127}{RGB}{127,127,127}
\definecolor{mediumpurple148103189}{RGB}{148,103,189}
\definecolor{orchid227119194}{RGB}{227,119,194}
\definecolor{sienna1408675}{RGB}{140,86,75}
\definecolor{steelblue31119180}{RGB}{31,119,180}

\begin{axis}[
width=1.0*\columnwidth, 
height=3.5cm,
tick align=inside,
tick pos=left,
x grid style={darkgray176},
% xmajorgrids,
xmin=-0.1796, xmax=2.8916,
xmajorticks=false,
y grid style={darkgray176},
% ymajorgrids,
ylabel={\footnotesize Success Prob.},
ymin=0, ymax=1.05,
ytick style={color=black},
ytick={0.0, 0.1, 0.2, 0.3, 0.4, 0.5, 0.6, 0.7, 0.8, 0.9, 1.0},
yticklabels={0.0, , 0.2, ,0.4, ,0.6, , 0.8, , 1.0},
]

\draw[draw=minDependencyPolicy, pattern=north west lines, pattern color=minDependencyPolicy] (axis cs:-0.04,0) rectangle (axis cs:0.04,0.866);
\draw[draw=minDependencyPolicy, pattern=north west lines, pattern color=minDependencyPolicy] (axis cs:0.96,0) rectangle (axis cs:1.04,0.876);
\draw[draw=minDependencyPolicy, pattern=north west lines, pattern color=minDependencyPolicy] (axis cs:1.96,0) rectangle (axis cs:2.04,0.884);

\draw[draw=baselinePolicy, pattern=north west lines, pattern color=baselinePolicy] (axis cs:0.384,0) rectangle (axis cs:0.464,0.344);
\draw[draw=baselinePolicy, pattern=north west lines, pattern color=baselinePolicy] (axis cs:1.384,0) rectangle (axis cs:1.464,0.238);
\draw[draw=baselinePolicy, pattern=north west lines, pattern color=baselinePolicy] (axis cs:2.384,0) rectangle (axis cs:2.464,0.345);

\draw[draw=minDependencyPolicy, pattern=grid, pattern color=minDependencyPolicy] (axis cs:0.056,0) rectangle (axis cs:0.136,0.93);
\draw[draw=minDependencyPolicy, pattern=grid, pattern color=minDependencyPolicy] (axis cs:1.056,0) rectangle (axis cs:1.136,0.932);
\draw[draw=minDependencyPolicy, pattern=grid, pattern color=minDependencyPolicy] (axis cs:2.056,0) rectangle (axis cs:2.136,0.937);

\draw[draw=baselinePolicy, pattern=grid, pattern color=baselinePolicy] (axis cs:0.48,0) rectangle (axis cs:0.56,0.378);
\draw[draw=baselinePolicy, pattern=grid, pattern color=baselinePolicy] (axis cs:1.48,0) rectangle (axis cs:1.56,0.242);
\draw[draw=baselinePolicy, pattern=grid, pattern color=baselinePolicy] (axis cs:2.48,0) rectangle (axis cs:2.56,0.348);

\draw[draw=minDependencyPolicy, pattern=crosshatch, pattern color=minDependencyPolicy] (axis cs:0.152,0) rectangle (axis cs:0.232,1);
\draw[draw=minDependencyPolicy, pattern=crosshatch, pattern color=minDependencyPolicy] (axis cs:1.152,0) rectangle (axis cs:1.232,0.991);
\draw[draw=minDependencyPolicy, pattern=crosshatch, pattern color=minDependencyPolicy] (axis cs:2.152,0) rectangle (axis cs:2.232,0.996);

\draw[draw=baselinePolicy, pattern=crosshatch, pattern color=baselinePolicy] (axis cs:0.576,0) rectangle (axis cs:0.656,0.505);
\draw[draw=baselinePolicy, pattern=crosshatch, pattern color=baselinePolicy] (axis cs:1.576,0) rectangle (axis cs:1.656,0.223);
\draw[draw=baselinePolicy, pattern=crosshatch, pattern color=baselinePolicy] (axis cs:2.576,0) rectangle (axis cs:2.656,0.413);

\draw[draw=minDependencyPolicy, fill=minDependencyPolicy] (axis cs:0.248,0) rectangle (axis cs:0.328,0.99973651515644);
\draw[draw=minDependencyPolicy,fill=minDependencyPolicy] (axis cs:1.248,0) rectangle (axis cs:1.328,0.999744877312207);
\draw[draw=minDependencyPolicy,fill=minDependencyPolicy] (axis cs:2.248,0) rectangle (axis cs:2.328,0.992880614958785);

\draw[draw=baselinePolicy,fill=baselinePolicy] (axis cs:0.672,0) rectangle (axis cs:0.752,0.999999999999297);
\draw[draw=baselinePolicy,fill=baselinePolicy] (axis cs:1.672,0) rectangle (axis cs:1.752,0.999999999999605);
\draw[draw=baselinePolicy,fill=baselinePolicy] (axis cs:2.672,0) rectangle (axis cs:2.752,0.999999999998752);
\end{axis}

% Label the bars with the corresponding policies
\node(init1MDPolicy) [] at (0.72cm, -0.2cm) {\(\jointPolicy_{MD}\)};
\node(init1BasePolicy) [right=-0.05cm of init1MDPolicy] {\(\jointPolicy_{base}\)};

\node(init2MDPolicy) [right=0.25cm of init1BasePolicy] {\(\jointPolicy_{MD}\)};
\node(init2BasePolicy) [right=-0.05cm of init2MDPolicy] {\(\jointPolicy_{base}\)};

\node(init3MDPolicy) [right=0.25cm of init2BasePolicy] {\(\jointPolicy_{MD}\)};
\node(init3BasePolicy) [right=-0.05cm of init3MDPolicy] {\(\jointPolicy_{base}\)};

% Put in the initial configuration labels.
\draw [decorate, decoration = {calligraphic brace, mirror}, very thick] (0.2cm, -0.4cm) --  (2.2cm, -0.4cm);
\draw [decorate, decoration = {calligraphic brace, mirror}, very thick] (2.5cm, -0.4cm) --  (4.5cm, -0.4cm);
\draw [decorate, decoration = {calligraphic brace, mirror}, very thick] (4.8cm, -0.4cm) --  (6.8cm, -0.4cm);

\node(init1Label) [font=\footnotesize, align=center, text width=2cm] at (1.2cm, -1.0cm) {One server in need of repair.};
\node(init2Label) [font=\footnotesize, align=center, text width=2cm] at (3.5cm, -1.0cm) {Two servers in repair.};
\node(init3Label) [font=\footnotesize, align=center, text width=2cm] at (5.8cm, -1.0cm) {Two in repair, two offline.};

\node(labelForLabels) [font=\footnotesize, align=center, text width=1.5cm] at (-0.7cm, -1.0cm) {Initial\\config:};

% \node(policy_base) [below=0.8cm of agent1Reachability, xshift=1.6cm] {\(\policy_{base}\)};

\begin{customlegend}[
    legend columns=4, 
    legend style={
        align=left, 
        column sep=0.7ex, 
        font=\footnotesize, 
        % draw=white!50!black,
        draw=none,
        fill=white,
        fill opacity=0.8,
        rounded corners=1mm,
        at={(73.0mm, 26.0mm)},
    }, 
    legend entries={\(\privacyLevel = 0.1\), \(\privacyLevel = 1.0\), \(\privacyLevel = 10.0\), No privacy}
]
% \addlegendimage{ultra thick, minDependencyPolicy, mark=*, mark size=1, mark options={solid}}
\addlegendimage{area legend,pattern=north west lines, pattern color=gray}
\addlegendimage{area legend,pattern=grid, pattern color=gray}
\addlegendimage{area legend,pattern=crosshatch, pattern color=gray}
\addlegendimage{area legend, fill=gray}
\end{customlegend}

\end{tikzpicture}
    \vspace{-0.5cm}
    \caption{
    Probability of success in the SysAdmin example under a variety of initial system configurations and privacy levels.
    }
    \label{fig:sys_admin_bar_chart}
\end{figure}

\subsection{Additional Discussion}
\label{subsec:additional_experimental_discussion}

In addition to the differences between the values of the team's probability of success under \(\jointPolicy_{MD}\) and \(\jointPolicy_{base}\), we also observe a significant change in the expected length of the trajectories that result from these policies.
For example, under truthful communication in the SysAdmin experiments, the expected length of the trajectories induced by \(\jointPolicy_{base}\) range from \(30\) to \(40\) timesteps, depending on the initial configuration of the system. 
For \(\jointPolicy_{MD}\) these values range from \(3\) to \(6\) timesteps.

This observation gives insight into differences in the qualitative behaviors of the policies.
\(\jointPolicy_{base}\) induces conservative behavior that maximizes the team's probability of success by requiring the agents to wait for specific joint states before taking certain actions; the actions of each agent are highly dependent on the exact states of its teammates.
On the other hand, \(\jointPolicy_{MD}\) achieves nearly the same probability of success as \(\jointPolicy_{base}\), but the agents act quickly and accept a small level of risk in order to reduce the dependencies of their actions on the states of their teammates.

The inclusion of the total correlation as a regularization term prevents the policy synthesis procedure from making the agents highly interdependent in order to achieve a marginally higher probability of success.
This tradeoff becomes highly relevant when the inter-agent communications are imperfect, which is necessary in privatized multiagent systems.

Finally, we remark that in some settings there may not exist a collection of highly independent policies that achieve a high performance.
In such cases, we may not observe a large of a gap in performance between \(\jointPolicy_{MD}\) and \(\jointPolicy_{base}\) under private communication.
However, even in these settings, the value of the total correlation may act as an indicator that it is infeasible to achieve strong performance and privacy simultaneously.

\section{Conclusions}
\label{sec:conclusions}

This paper presents a framework to privatize inter-agent communications in cooperative multiagent decision-making problems. Specifically, we adopt a differential privacy mechanism to protect the symbolic state trajectories of agents. We provide theoretical results to analyze the tradeoff between the strength of privacy and the team’s performance. We synthesize robust policies for agents by reducing the total correlation among them. Numerical results demonstrate that the minimum-dependency policies achieve high performance under strong levels of privacy, whereas the team performance of baseline policies that ignore total correlation decreases dramatically under private communications.  

\section*{Acknowledgments}
This work was supported in part by AFRL FA9550-19-1-0169, AFRL FA8651-23-F-A008, ARL ACC-APG-RTP W911NF1920333, ARO W911NF-20-1-0140, NASA 80NSSC21M0071, NSF 1943275, and ONR N00014-21-1-2502.

\bibliographystyle{named}
\bibliography{bibliography}

\begin{thebibliography}{}

\bibitem[\protect\citeauthoryear{Becker \bgroup \em et al.\egroup
  }{2003}]{becker2003transition}
Raphen Becker, Shlomo Zilberstein, Victor Lesser, and Claudia~V Goldman.
\newblock Transition-independent decentralized {Markov} decision processes.
\newblock In {\em Proceedings of the 2nd International Conference on Autonomous
  Agents and Multiagent Systems}, pages 41--48, 2003.

\bibitem[\protect\citeauthoryear{Brafman}{2015}]{brafman2015privacy}
Ronen~Israel Brafman.
\newblock A privacy preserving algorithm for multi-agent planning and search.
\newblock In {\em Twenty-Fourth International Joint Conference on Artificial
  Intelligence}, 2015.

\bibitem[\protect\citeauthoryear{Bretagnolle and
  Huber}{1979}]{bretagnolle1979estimation}
Jean Bretagnolle and Catherine Huber.
\newblock Estimation des densit{\'e}s: risque minimax.
\newblock {\em Zeitschrift f{\"u}r Wahrscheinlichkeitstheorie und verwandte
  Gebiete}, 47(2):119--137, 1979.

\bibitem[\protect\citeauthoryear{Cao \bgroup \em et al.\egroup
  }{2013}]{Cao2013distributed}
Yongcan Cao, Wenwu Yu, Wei Ren, and Guanrong Chen.
\newblock An overview of recent progress in the study of distributed
  multi-agent coordination.
\newblock {\em IEEE Transactions on Industrial Informatics}, 9(1):427--438,
  2013.

\bibitem[\protect\citeauthoryear{Chen \bgroup \em et al.\egroup
  }{2022}]{chen2022differential}
Bo~Chen, Kevin Leahy, Austin Jones, and Matthew Hale.
\newblock Differential privacy for symbolic systems with application to markov
  chains.
\newblock {\em arXiv preprint arXiv:2202.03325}, 2022.

\bibitem[\protect\citeauthoryear{Choudhury \bgroup \em et al.\egroup
  }{2021}]{choudhury2021scalable}
Shushman Choudhury, Jayesh~K Gupta, Peter Morales, and Mykel~J Kochenderfer.
\newblock Scalable anytime planning for multi-agent mdps.
\newblock {\em arXiv preprint arXiv:2101.04788}, 2021.

\bibitem[\protect\citeauthoryear{Cover and Thomas}{1991}]{thomas2006elements}
Thomas~M Cover and Joy~A Thomas.
\newblock {\em Elements of Information Theory}.
\newblock John Wiley \& Sons, New York, 1991.

\bibitem[\protect\citeauthoryear{Dwork \bgroup \em et al.\egroup
  }{2014}]{dwork2014algorithmic}
Cynthia Dwork, Aaron Roth, et~al.
\newblock The algorithmic foundations of differential privacy.
\newblock {\em Foundations and Trends{\textregistered} in Theoretical Computer
  Science}, 9(3--4):211--407, 2014.

\bibitem[\protect\citeauthoryear{Garcelon \bgroup \em et al.\egroup
  }{2020}]{garcelon2020local}
Evrard Garcelon, Vianney Perchet, Ciara Pike-Burke, and Matteo Pirotta.
\newblock Local differential privacy for regret minimization in reinforcement
  learning.
\newblock {\em arXiv preprint arXiv:2010.07778}, 2020.

\bibitem[\protect\citeauthoryear{Gohari \bgroup \em et al.\egroup
  }{2021}]{gohari2021differential}
Parham Gohari, Bo~Wu, Calvin Hawkins, Matthew Hale, and Ufuk Topcu.
\newblock Differential privacy on the unit simplex via the dirichlet mechanism.
\newblock {\em IEEE Transactions on Information Forensics and Security},
  16:2326--2340, 2021.

\bibitem[\protect\citeauthoryear{Guestrin \bgroup \em et al.\egroup
  }{2003}]{guestrin2003efficient}
Carlos Guestrin, Daphne Koller, Ronald Parr, and Shobha Venkataraman.
\newblock Efficient solution algorithms for factored mdps.
\newblock {\em Journal of Artificial Intelligence Research}, 19:399--468, 2003.

\bibitem[\protect\citeauthoryear{Hefner \bgroup \em et al.\egroup
  }{2022}]{hefner2022privacy}
Tommy Hefner, Guy Shani, and Roni Stern.
\newblock Privacy preserving planning in multi-agent stochastic environments.
\newblock {\em Autonomous Agents and Multi-Agent Systems}, 36(1):1--27, 2022.

\bibitem[\protect\citeauthoryear{Hsu \bgroup \em et al.\egroup
  }{2014}]{Hsu2014DifferentialPA}
Justin Hsu, Marco Gaboardi, Andreas Haeberlen, Sanjeev Khanna, Arjun Narayan,
  Benjamin~C. Pierce, and Aaron Roth.
\newblock Differential privacy: An economic method for choosing epsilon.
\newblock {\em 2014 IEEE 27th Computer Security Foundations Symposium}, pages
  398--410, 2014.

\bibitem[\protect\citeauthoryear{Karabag \bgroup \em et al.\egroup
  }{2022}]{Karabag2022}
Mustafa~O. Karabag, Cyrus Neary, and Ufuk Topcu.
\newblock Planning not to talk: Multiagent systems that are robust to
  communication loss.
\newblock In {\em Proceedings of the 21st International Conference on
  Autonomous Agents and Multiagent Systems}, AAMAS '22, page 705–713,
  Richland, SC, 2022. International Foundation for Autonomous Agents and
  Multiagent Systems.

\bibitem[\protect\citeauthoryear{Lanckriet and
  Sriperumbudur}{2009}]{Lanckriet2009convergence}
Gert Lanckriet and Bharath~K. Sriperumbudur.
\newblock On the convergence of the concave-convex procedure.
\newblock In Y.~Bengio, D.~Schuurmans, J.~Lafferty, C.~Williams, and
  A.~Culotta, editors, {\em Advances in Neural Information Processing Systems},
  volume~22. Curran Associates, Inc., 2009.

\bibitem[\protect\citeauthoryear{Nissim and
  Brafman}{2014}]{nissim2014distributed}
Raz Nissim and Ronen Brafman.
\newblock Distributed heuristic forward search for multi-agent planning.
\newblock {\em Journal of Artificial Intelligence Research}, 51:293--332, 2014.

\bibitem[\protect\citeauthoryear{Oliehoek and
  Amato}{2016}]{oliehoek2016concise}
Frans~A Oliehoek and Christopher Amato.
\newblock {\em A concise introduction to decentralized {POMDPs}}.
\newblock Springer, Cham, 2016.

\bibitem[\protect\citeauthoryear{Parker \bgroup \em et al.\egroup
  }{2016}]{Parker2016robot}
Lynne~E. Parker, Daniela Rus, and Gaurav~S. Sukhatme.
\newblock {\em Multiple Mobile Robot Systems}, pages 1335--1384.
\newblock Springer International Publishing, Cham, 2016.

\bibitem[\protect\citeauthoryear{Qiao and Wang}{2022}]{qiao2022offline}
Dan Qiao and Yu-Xiang Wang.
\newblock Offline reinforcement learning with differential privacy.
\newblock {\em arXiv preprint arXiv:2206.00810}, 2022.

\bibitem[\protect\citeauthoryear{Rashid \bgroup \em et al.\egroup
  }{2018}]{rashid2018qmix}
Tabish Rashid, Mikayel Samvelyan, Christian Schroeder, Gregory Farquhar, Jakob
  Foerster, and Shimon Whiteson.
\newblock Qmix: Monotonic value function factorisation for deep multi-agent
  reinforcement learning.
\newblock In {\em Proceedings of the 35th International Conference on Machine
  Learning}, pages 4295--4304, 2018.

\bibitem[\protect\citeauthoryear{Schulz and Mihov}{2003}]{Schulz2003distance}
Klaus~U. Schulz and Stoyan Mihov.
\newblock Fast string correction with levenshtein automata.
\newblock {\em International Journal on Document Analysis and Recognition
  (IJDAR)}, 5(1):67--85, 2003.

\bibitem[\protect\citeauthoryear{Son \bgroup \em et al.\egroup
  }{2019}]{son2019qtran}
Kyunghwan Son, Daewoo Kim, Wan~Ju Kang, David~Earl Hostallero, and Yung Yi.
\newblock Qtran: Learning to factorize with transformation for cooperative
  multi-agent reinforcement learning.
\newblock In {\em Proceedings of the 36th International Conference on Machine
  Learning}, pages 5887--5896, 2019.

\bibitem[\protect\citeauthoryear{Such \bgroup \em et al.\egroup
  }{2014}]{such2014privacysurvay}
Jose~M. Such, Agustín Espinosa, and Ana García-Fornes.
\newblock A survey of privacy in multi-agent systems.
\newblock {\em The Knowledge Engineering Review}, 29(3):314–344, 2014.

\bibitem[\protect\citeauthoryear{Ye \bgroup \em et al.\egroup }{2022}]{Ye2022}
Dayong Ye, Tianqing Zhu, Sheng Shen, Wanlei Zhou, and Philip~S. Yu.
\newblock Differentially private multi-agent planning for logistic-like
  problems.
\newblock {\em IEEE Transactions on Dependable and Secure Computing},
  19(2):1212--1226, 2022.

\bibitem[\protect\citeauthoryear{Yuille and
  Rangarajan}{2001}]{Yuille2001procedure}
Alan~L Yuille and Anand Rangarajan.
\newblock The concave-convex procedure (cccp).
\newblock In T.~Dietterich, S.~Becker, and Z.~Ghahramani, editors, {\em
  Advances in Neural Information Processing Systems}, volume~14. MIT Press,
  2001.

\end{thebibliography}

\onecolumn

\begin{center}
\textbf{\LARGE Differential Privacy in Cooperative Multiagent Planning: Supplementary Material}
\end{center}

\appendix

\section{Proofs for Theoretical Results}

The Kullback-Leibler (KL) divergence \cite{thomas2006elements} between discrete probability distributions $Q^1$ and $Q^2$ with supports $\bm{Q}^1$ and $\bm{Q}^2$, respectively, is
\[
K L\left(Q^1 \| Q^2\right)=\sum_{q \in \bm{Q}^1} Q^1(q) \log \left(\frac{Q^1(q)}{Q^2(q)}\right).
\]

\emph{Notations} We first define some notations that will be used for the proofs. Let $\gameStateRandomVar_t$ be a random variable denoting the joint state of the agents at time $t$ under the joint policy with no privatization, $\gameActionRandomVar_t$ be a random variable denoting the joint action of the agents at time $t$, $\mdpStateRandomVar_t^i$ be a random variable denoting the state of agent $i$ at time $t$, $\mdpActionRandomVar_t^i$ be a random variable denoting the action of Agent $i$ at time $t$. $\gameStateRandomVar_t^{-i}$ be a random variable denoting the state of agent $i$'s teammate exclude agent $i$ itself at time $t$, and $\gameActionRandomVar_t^{-i}$ be a random variable denoting the action of agent $i$'s teammate exclude agent $i$ itself at time $t$.
The total correlation $\totalCorrelation_{\jointPolicy}$ of a joint policy $\jointPolicy$ is
\begin{equation}
    \totalCorrelation_{\jointPolicy} = \sum_{i=1}^\numAgents \entropy(\mdpStateRandomVar_0^{i}\mdpActionRandomVar_0^{i}\ldots\mdpStateRandomVar_{\randomReachTime}^{i})-\entropy(\gameStateRandomVar_0\gameActionRandomVar_0\ldots\gameStateRandomVar_\randomReachTime)%\label{eq:proof_tc_def}
\end{equation}
where \(\randomReachTime\) denotes the random hitting time to \(\targetSet \cup \deadSetPrime\), i.e., the effective end of the trajectory in terms of the reach-avoid specification~\cite{Karabag2022}.  

Let \(\pathSet\) denote all trajectory fragments that end at a state in \(\targetSet \cup \deadSetPrime\), i.e., \(\pathSet = \lbrace \genericPath=\bm{s_0}\bm{a_0}\ldots\bm{s_T}|\bm{s_T} \in  \targetSet \cup \deadSetPrime \text{ and } \forall t < T,  \bm{s_t} \not\in \targetSet \cup \deadSetPrime \rbrace \), and \(\pathSet'\) denote all trajectories that never reach \(\targetSet \cup \deadSetPrime\), i.e., \(\pathSet' = \lbrace w=\genericPath=\bm{s_0}\bm{a_0}\ldots| \forall t \geq 0,  \bm{s_t} \not\in \targetSet \cup \deadSetPrime \rbrace \). Note that every trajectory either starts with a trajectory fragment from \(\pathSet\) or is in \(\pathSet'\). Also, let \(R \subseteq \pathSet \cup \pathSet'\) denote all trajectory fragments that end at a state in \(\targetSet\), i.e., \(R = \lbrace \genericPath=\bm{s_0}\bm{a_0}\ldots\bm{s_T}|\bm{s_T} \in  \targetSet  \text{ and } \forall t < T,  \bm{s_t} \not\in \targetSet \cup \deadSet  \rbrace \).

Let $\Gamma^{\oblivious}$ be the distribution of joint trajectories induced by the joint policy executed with truthful communications (i.e., no privacy). Also, let $\Gamma^{\private}$ be the distribution of joint trajectories with privacy enforced. Let $v^{\oblivious}$ be the probability of success under truthful communications and $v^{\private}$ be the probability of success under private communications.

We use $\probabilityMeasure^{\oblivious}$ to denote the probability measure over the actual (finite or infinite) state-action process under the joint policy with truthful communications. $\probabilityMeasure^{\private}$ denotes the probability measure over the actual (finite or infinite) state-action process under joint policy with private communications. With abuse of notation, we also use $\probabilityMeasure_{\epsilon}$ to denote the conditional probability measure over private state trajectories given the actual state trajectory.

Let $\genericPath=\gameState_0\gameAction_0\gameState_1\gameAction_1\ldots\gameState_T\in (\bm{\mdpStateSpace}\times\bm{\mdpActionSpace})^T$ be a joint trajectory fragment and $\Tilde{\genericPath}=\Tilde{\gameState}_0\Tilde{\gameState}_1\ldots\Tilde{\gameState}_T\in \gameStateSpace^T$ be a private joint state trajectory fragment. We use $\bm{\hat{\mdpState}_t^j}=\{\Tilde{\mdpState}_t^0,\dots,\Tilde{\mdpState}_t^{j-1},\mdpState_t^j,\Tilde{\mdpState}_t^{j+1},\dots,\Tilde{\mdpState}_t^{N}\}$ to denote agent $j$'s copy of private joint state.

The Kleene star applied to a set $V$ of symbols is the set $V^*=\cup_{i\geq 0}V^i$ of all finite-length words where $V^0=\{\Lambda\}$ and $\Lambda$ is the empty string. The set of all infinite-length words is denoted by $V^\omega$.

We introduce the following lemma, to be used in the proof of other theoretical results.

\begin{lemma} \label{lemma:expectedlogproboftrueword}
    \begin{align}
        &\expectation_{\genericPath \sim \probabilityMeasure^{\oblivious}}\left[\log \probabilityMeasure_{\privacyLevel}(\privatePath = \genericPath|\genericPath) \right] \geq  - \numAgents\log\left(\left(\outdegree_{\max}-1\right)\exp(-\frac{\privacyLevel}{\adjParam})+1\right)l^{\oblivious}
    \end{align} 
    where $\outdegree(\mdpState_{t-1}^{i})$ is the out degree of $\mdpState_{t-1}^{i}$ and $\outdegree_{\max}=\max_{\mdpState\in \cup_{i=1}^{\numAgents} \mdpStateSpace^{i}}  \outdegree(\mdpState)$.
\end{lemma}

\begin{proof}[Proof of Lemma \ref{lemma:expectedlogproboftrueword}]

Due to the Markovianity of the online privacy mechanism (Algorithm \ref{alg:privacy_construction}) and independence between the agents, we have  
\begin{align}
     \probabilityMeasure_{\privacyLevel}(\privatePath=\genericPath|\genericPath=\gameState_{0}\ldots \gameState_{\timeHorizon}) = \prod_{t=0}^{\timeHorizon-1}\prod_{i=1}^\numAgents\probabilityMeasure_\privacyLevel(\Tilde{\mdpState}_t^i = \mdpState_t^i|\mdpState_t^i,\Tilde{\mdpState}_{t-1}^i)
\end{align}
We note that if \(\privatePath=\genericPath\), then for all \(t\geq 0\) and \(j \in [\numAgents]\), we have \(\hat{\gameState}_{t}^j = \gameState_{t}\),  i.e., the copy of the private state for every agent always matches the actual joint state. Hence, 
\begin{align}
    \probabilityMeasure_{\privacyLevel}(\privatePath=\genericPath|\genericPath=\gameState_{0}\ldots \gameState_{\timeHorizon}) = \prod_{t=0}^{\timeHorizon-1}\prod_{i=1}^{\numAgents}\probabilityMeasure_\privacyLevel(\Tilde{\mdpState}_t^i = \mdpState_t^i|\mdpState_t^i,\Tilde{\mdpState}_{t-1}^i = \mdpState_{t-1}^i) 
\end{align}

From~\cite[Theorem 7]{chen2022differential}, we have
\begin{align}
\probabilityMeasure_\privacyLevel(\Tilde{\mdpState}_t^i = \mdpState_t^i|\mdpState_t^i,\Tilde{\mdpState}_{t-1}^i = \mdpState_{t-1}^i)=\frac{1}{\left(\outdegree(\mdpState_{t-1}^{i})-1\right)\exp(-\frac{\privacyLevel}{\adjParam})+1}, 
\end{align}
where $\outdegree(\mdpState_{t-1}^{i})$ is the out degree of $\mdpState_{t-1}^{i}.$  Let $\outdegree_{\max}=\max_{\mdpState\in \cup_{i=1}^{\numAgents} \mdpStateSpace^{i}}  \outdegree(\mdpState)$ which gives $$\probabilityMeasure_\privacyLevel(\Tilde{\mdpState}_t^i = \mdpState_t^i|\mdpState_t^i,\Tilde{\mdpState}_{t-1}^i = \mdpState_{t-1}^i)\geq\frac{1}{\left(\outdegree_{\max}-1\right)\exp(-\frac{\privacyLevel}{\adjParam})+1}.$$ Using this, we get
\begin{align*}
\log \probabilityMeasure_{\privacyLevel}(\privatePath=\genericPath|\genericPath=\mdpState_{0}\ldots \mdpState_{\timeHorizon})
&=\log\left(\prod_{t=0}^{\timeHorizon-1}\prod_{i=1}^\numAgents\probabilityMeasure_\privacyLevel(\mdpState_t^i|\mdpState_t^i,\Tilde{\mdpState}_{t-1}^i = \mdpState_{t-1}^i)\right)
\\
&=\sum_{t=0}^{\timeHorizon-1}\sum_{i=1}^\numAgents\log\probabilityMeasure_\privacyLevel(\mdpState_t^i|\mdpState_t^i,\Tilde{\mdpState}_{t-1}^i = \mdpState_{t-1}^i)
\\
&\geq\sum_{t=0}^{\timeHorizon-1}\sum_{i=1}^\numAgents\log\frac{1}{\left(\outdegree_{\max}-1\right)\exp(-\frac{\privacyLevel}{\adjParam})+1}
\\
&=\sum_{t=0}^{{\timeHorizon-1}}-\numAgents\log\left(\left(\outdegree_{\max}-1\right)\exp(-\frac{\privacyLevel}{\adjParam})+1\right).
\end{align*}
Consequently,
\begin{align*}\expectation_{\genericPath \sim \probabilityMeasure^{\oblivious}}\left[\log \probabilityMeasure_{\privacyLevel}(\privatePath = \genericPath|\genericPath) \right]
&=
\sum_{\genericPath\in W}\probabilityMeasure^{\oblivious}(\genericPath)\log\left(\prod_{t=0}^{\timeHorizon-1}\prod_{i=1}^\numAgents\probabilityMeasure_\privacyLevel(\mdpState_t^i|\mdpState_t^i,\mdpState_{t-1}^i)\right)
\\
    &\geq    \sum_{\genericPath\in  W} -\probabilityMeasure^{\oblivious}(\genericPath)\sum_{t=0}^{\timeHorizon-1}\numAgents\log\left(\left(\outdegree_{\max}-1\right)\exp(-\frac{\privacyLevel}{\adjParam})+1\right)
    \\
    &= E\left[\sum_{t=0}^{\tau -1}-\numAgents\log\left(\left(\outdegree_{\max}-1\right)\exp(-\frac{\privacyLevel}{\adjParam})+1\right)\lvert\probabilityMeasure^{\oblivious}\right]
    \\
    &=-\numAgents\log\left(\left(\outdegree_{\max}-1\right)\exp(-\frac{\privacyLevel}{\adjParam})+1\right)E\left[\sum_{t=0}^{\tau -1}1\lvert\probabilityMeasure^{\oblivious}\right]
    \\
    &=-\numAgents\log\left(\left(\outdegree_{\max}-1\right)\exp(-\frac{\privacyLevel}{\adjParam})+1\right)l^{\oblivious}. 
    % \label{eq:prob_true_word}
\end{align*}

\end{proof}

\begin{proof} [Proof of Theorem \ref{thm:performance_bound}]
Due to the causality property of the only mechanism (Algorithm \ref{alg:privacy_construction}) and the joint policy execution (Algorithm  \ref{alg:policy_exec}), we have
\begin{align}
    \probabilityMeasure^{\private}(\genericPath)&= \sum_{\privatePath \in \gameStateSpace^T} \probabilityMeasure^{\private}(\genericPath,\privatePath)\nonumber\\
    &= \sum_{\privatePath\in \gameStateSpace^T}\prod_{t=0}^{T-1}\Pr(\gameAction_t\gameState_{t+1},\privateGameState_t|\gameAction_{t-1}\gameState_t\ldots\gameAction_0\gameState_1,\privateGameState_{t-1}\ldots\privateGameState_0),\nonumber
\end{align}
where,
\begin{align}\Pr(\gameAction_t\gameState_{t+1},\privateGameState_t|\gameAction_{t-1}\gameState_t \ldots\gameAction_0\gameState_1,\privateGameState_{t-1}\ldots\privateGameState_0)
    &= \Pr(\gameAction_t\gameState_{t+1},\privateGameState_t|\gameState_t,\privateGameState_{t-1})\label{eq:proof_1}\\
    &=  \prod_{i=1}^\numAgents\Pr(\mdpAction_{t}^i\mdpState_{t+1}^i,\privateGameState_t|\gameState_t,\privateGameState_{t-1})\label{eq:proof_2} \\
    &= \prod_{i=1}^\numAgents  \Pr(\mdpAction_{t}^i\mdpState_{t+1}^i|\gameState_t,\privateGameState_{t},\privateGameState_{t-1})\Pr(\privateGameState_{t}|\gameState_{t},\privateGameState_{t-1})\label{eq:proof_4}\\
    &= \prod_{i=1}^\numAgents \Pr(\mdpAction_{t}^i\mdpState_{t+1}^i|\gameState_{t},\privateGameState_t,\privateGameState_{t-1})\left(\prod_{k=1}^\numAgents\probabilityMeasure_\epsilon(\Tilde{s}_t^k|\mdpState_t^k,\Tilde{s}_{t-1}^k)\right)\label{eq:proof_5}\\
    &= \prod_{i=1}^\numAgents \Pr(\mdpAction_{t}^i\mdpState_{t+1}^i|\hat{\gameState}_t^i)\left(\prod_{k=1}^\numAgents\probabilityMeasure_\epsilon(\Tilde{\mdpState}_t^k|\mdpState_t^k,\Tilde{\mdpState}_{t-1}^k)\right)\label{eq:proof_6}\\
    &= \prod_{i=1}^\numAgents \mdpTransition(\mdpState_t^i,\mdpAction_{t}^i,\mdpState_{t+1}^i)\localPolicy^{i}(\hat{\gameState}_t^i, \mdpAction_{t}^i)\left(\prod_{k=1}^\numAgents\probabilityMeasure_\epsilon(\Tilde{\mdpState}_t^k|\mdpState_t^k,\Tilde{\mdpState}_{t-1}^k)\right)\label{eq:proof_7}.
\end{align}

Equation~\eqref{eq:proof_1} is because of the Markovian property. Equation~\eqref{eq:proof_2} is because the each agent are choosing its next action and state independently. Equation~\eqref{eq:proof_5} is due to each state is generating its private state independently. Equation~\eqref{eq:proof_6} is because for each agent $i$, its true next state $\Tilde{s}_{t+1}^i$ is independent of other states' true states and the private state itself. 

Therefore, 
\begin{align}
    \probabilityMeasure^{\private}(\genericPath)
    &=\sum_{\privatePath\in \gameStateSpace^{T}}\prod_{t=0}^{T-1}\prod_{i=1}^\numAgents \mdpTransition(\mdpState_t^i,\mdpAction_{t}^i,\mdpState_{t+1}^i)\localPolicy^{i}(\hat{\mathbf{s}}_t^i, \mdpAction_{t}^i)\left(\prod_{i=1}^\numAgents\probabilityMeasure_\epsilon(\Tilde{s}_t^i|\mdpState_t^i,\Tilde{s}_{t-1}^i)\right)
    \nonumber\\
    &\geq \prod_{t=0}^{T-1}\prod_{i=1}^\numAgents \mdpTransition(\mdpState_t^i,\mdpAction_{t}^i,\mdpState_{t+1}^i)\localPolicy^{i}(\hat{\mathbf{s}}_t^i, \mdpAction_{t}^i)\left(\prod_{i=1}^\numAgents\probabilityMeasure_\epsilon(\Tilde{s}_t^i|\mdpState_t^i,\Tilde{s}_{t-1}^i)\right), \forall \privatePath\in \gameStateSpace^T,
     \label{eq:proof_8}
\end{align}
where Equation~\eqref{eq:proof_8} is because the probability of all possible private state trajectories has to be greater than any single private state trajectory. We only consider the case when $\privateGameState_t=\gameState_{t}$, which means the private online mechanism will make the correct decision at every time $t$. Therefore,
\begin{align}
    \probabilityMeasure^{\private}(\genericPath) &\geq \prod_{t=0}^{T-1}\prod_{i=1}^\numAgents \mdpTransition(\mdpState_t^i,\mdpAction_{t}^i,\mdpState_{t+1}^i)\localPolicy^{i}(\gameState_{t},\mdpAction_{t}^i)\left(\prod_{i=1}^\numAgents\probabilityMeasure_\epsilon(\mdpState_t^i|\mdpState_t^i,\mdpState_{t-1}^i)\right)
    \\
    &=\probabilityMeasure^{\oblivious}(\genericPath)\left(\prod_{i=1}^\numAgents\probabilityMeasure_\epsilon(\mdpState_t^i|\mdpState_t^i,\mdpState_{t-1}^i)\right) \label{eq:privfullrel} 
\end{align}

Now we look at the following KL divergence:
\begin{align}
KL(\gamePathDist^{\oblivious}||\gamePathDist^{\private})
    &=\sum_{\genericPath\in \pathSet \cup \pathSet'}\probabilityMeasure^{\oblivious}(\genericPath)\log\left(\frac{\probabilityMeasure^{\oblivious}(\genericPath)}{\probabilityMeasure^{\private}(\genericPath)}\right)\nonumber
    \\
    &=\sum_{\genericPath\in \pathSet}\probabilityMeasure^{\oblivious}(\genericPath)\log\left(\frac{\probabilityMeasure^{\oblivious}(\genericPath)}{\probabilityMeasure^{\private}(\genericPath)}\right) \label{eq:everypathends}
    \\
        &\leq \sum_{\genericPath\in \pathSet }\probabilityMeasure^{\oblivious}(\genericPath)\log\left(\frac{\probabilityMeasure^{\oblivious}(\genericPath)}{\probabilityMeasure^{\oblivious}(\genericPath)\left(\prod_{t=0}^\infty\prod_{i=1}^\numAgents\probabilityMeasure_\epsilon(\mdpState_t^i|\mdpState_t^i,\mdpState_{t-1}^i)\right)}\right)\label{eq:proof_12}
        \\
    &= \sum_{\genericPath\in \pathSet}\probabilityMeasure^{\oblivious}(\genericPath)\log(\probabilityMeasure^{\oblivious}(\genericPath)) - \sum_{\genericPath\in \pathSet} \probabilityMeasure^{\oblivious}(\genericPath)\log(\probabilityMeasure^{\oblivious}(\genericPath))\nonumber
    \\
    &\qquad - \sum_{\genericPath\in \pathSet}\probabilityMeasure^{\oblivious}(\genericPath)\log\left(\prod_{t=0}^\infty\prod_{i=1}^\numAgents\probabilityMeasure_\epsilon(\mdpState_t^i|\mdpState_t^i,\mdpState_{t-1}^i)\right)\nonumber
    \\
    &= \entropy(\gameStateRandomVar_0\gameActionRandomVar_0\ldots\gameStateRandomVar_\randomReachTime) -\entropy(\gameStateRandomVar_0\gameActionRandomVar_0\ldots\gameStateRandomVar_\randomReachTime) -\sum_{\genericPath\in \pathSet} \probabilityMeasure^{\oblivious}(\genericPath)\log\left(\prod_{t=0}^{T-1}\prod_{i=1}^\numAgents\probabilityMeasure_\epsilon(\mdpState_t^i|\mdpState_t^i,\mdpState_{t-1}^i)\right) 
    \\
    &\leq \sum_{i=1}^\numAgents \entropy(\mdpStateRandomVar_0^i\mdpActionRandomVar_0^i\ldots\mdpStateRandomVar_\randomReachTime^i)-\entropy(\gameStateRandomVar_0\gameActionRandomVar_0\ldots\gameStateRandomVar_\randomReachTime) -  \sum_{\genericPath\in \pathSet} \probabilityMeasure^{\oblivious}(\genericPath)\log\left(\prod_{t=0}^{T-1}\prod_{i=1}^\numAgents\probabilityMeasure_\epsilon(\mdpState_t^i|\mdpState_t^i,\mdpState_{t-1}^i)\right)\label{eq:proof_13}
    \\
    &= \totalCorrelation_{\jointPolicy} - \sum_{\genericPath\in \pathSet}\probabilityMeasure^{\oblivious}(\genericPath)\log\left(\prod_{t=0}^{T-1}\prod_{i=1}^\numAgents\probabilityMeasure_\epsilon(\mdpState_t^i|\mdpState_t^i,\mdpState_{t-1}^i)\right)\label{eq:proof_14}
    \\
    &=  \totalCorrelation_{\jointPolicy}  - \expectation_{\genericPath \sim \probabilityMeasure^{\oblivious}}\left[\probabilityMeasure_{\privacyLevel}(\privatePath = \genericPath|\genericPath) \right]
\end{align}
where \eqref{eq:everypathends} is due to \(\sum_{\genericPath\in \pathSet'}  \probabilityMeasure^{\oblivious}(\genericPath) = 0\), \eqref{eq:proof_13} is due to the subadditivity of entropy, and \eqref{eq:proof_14} is due to the definition of \(\totalCorrelation_{\jointPolicy}\).

Using Lemma \ref{lemma:expectedlogproboftrueword} in \eqref{eq:proof_14} gives
\begin{equation}
\kl(\gamePathDist^{\oblivious}||\gamePathDist^{\private})
    \leq \totalCorrelation_{\jointPolicy} + \numAgents\log\left(\left(\rho_{\max}-1\right)\exp(-\frac{\epsilon}{\ell})+1\right)l^{\oblivious}. \label{eq:klbound}
\end{equation}

Finally, we show that \(    v^{\private} \geq v^{\oblivious} -1 + \exp(-\totalCorrelation_{\jointPolicy})\left((\rho_{max}-1)\exp\left(-\frac{\epsilon}{\ell}\right)+1\right)^{Nl^{\oblivious}}/2.\)  Let \(R' \subseteq \pathSet \cup \pathSet'\) be an arbitrary set.
\begin{subequations}
\begin{align}
    \gameValue^{\oblivious} - \gameValue^{\private}
    &=  \sum_{\genericPath \in R} \probabilityMeasure^{\oblivious}(\genericPath ) - \probabilityMeasure^{\private}(\genericPath )
    \\
    &\leq \left | \sum_{\genericPath \in R} \probabilityMeasure^{\oblivious}(\genericPath ) - \probabilityMeasure^{\private}(\genericPath )\right |
    \\
    &\leq \sup_{R'} \left |\sum_{\genericPath \in R'} \probabilityMeasure^{\oblivious}(\genericPath ) - \probabilityMeasure^{\private}(\genericPath )\right |
    \\
    &\leq \sqrt{1-\exp(-\kl(\gamePathDist^{\oblivious} || \gamePathDist^{\private}))} \label{bretagnollehuber}
\end{align}
\end{subequations}
where \eqref{bretagnollehuber} is due to Bretagnolle-Huber inequality~\cite{bretagnolle1979estimation}.
Rearranging the terms of \eqref{bretagnollehuber} and using \eqref{eq:klbound} yields to the desired result.

\end{proof}

We note that apart from Theorem~\ref{thm:performance_bound}, we can derive a tighter lower bound on $\mdpValue^{pr}$. 

\begin{theorem}\label{thm:performance_2_bound}
Given $\epsilon>0$, for $\numAgents$ agents, we have
\begin{equation}\label{eq:thm_2_equation}
    \mdpValue^{pr}\geq \mdpValue^{\oblivious}-1+\left(\left(\outdegree_{\max}-1\right)\exp(-\frac{\epsilon}{\adjParam})+1\right)^{\numAgents l^{\oblivious}}.
\end{equation}
\end{theorem}

\begin{proof} [Proof of Theorem \ref{thm:performance_2_bound}]
    
As shown in the proof of Theorem \ref{thm:performance_bound}, we have
\begin{align}
    \mdpValue^{\private}&= \sum_{\genericPath=\gameState_0\gameAction_0\gameState_1\gameAction_1\ldots\gameState_T \in \pathSet} \probabilityMeasure^{\private}(\genericPath)  \mathds{1}(\genericPath \in \reachPathSet)  \\
    &\geq \sum_{\genericPath=\gameState_0\gameAction_0\gameState_1\gameAction_1\ldots\gameState_T \in \pathSet} \probabilityMeasure^{\oblivious} (\genericPath) \mathds{1}(\genericPath \in \reachPathSet) \left(\prod_{t=0}^{T-1} \prod_{k=1}^\numAgents \mu_\epsilon(\mdpState_t^k|\mdpState_t^k,\mdpState_{t-1}^k)\right)
    \\
    &= \Pr(\genericPath \in \reachPathSet \wedge \privatePath = \genericPath | \genericPath \sim \probabilityMeasure^{\oblivious}, \privatePath \sim \probabilityMeasure^{\epsilon}(\cdot| \genericPath)).
\end{align}

By the union bound, we have 
\begin{align}
    \mdpValue^{\private}&\geq \Pr(\genericPath \in \reachPathSet | \genericPath \sim \probabilityMeasure^{\oblivious}) + \expectation_{\genericPath \sim \probabilityMeasure^{\oblivious}}\left[ \probabilityMeasure_{\privacyLevel}(\privatePath = \genericPath|\genericPath) \right] - 1 
    \\
    &= \mdpValue^{\oblivious} + \expectation_{\genericPath \sim \probabilityMeasure^{\oblivious}}\left[ \probabilityMeasure_{\privacyLevel}(\privatePath = \genericPath|\genericPath) \right] -1 \label{eq:unionbound}
\end{align}

Then with \[\expectation_{\genericPath \sim \probabilityMeasure^{\oblivious}}\left[\probabilityMeasure_{\privacyLevel}(\privatePath = \genericPath|\genericPath) \right]=\sum_{\genericPath \in \pathSet}\probabilityMeasure^{\oblivious}(\genericPath)\probabilityMeasure_{\privacyLevel}(\privatePath = \genericPath|\genericPath) \]
and Jensen's inequality,  we have
\begin{align}
    \expectation_{\genericPath \sim \probabilityMeasure^{\oblivious}}\left[\probabilityMeasure_{\privacyLevel}(\privatePath = \genericPath|\genericPath) \right]
    &=\exp\left(\log\sum_{\genericPath \in \pathSet}\probabilityMeasure^{\oblivious}(\genericPath)\probabilityMeasure_{\privacyLevel}(\privatePath = \genericPath|\genericPath) \right)
    \\
    &\geq \exp\left(\sum_{\genericPath \in \pathSet}\probabilityMeasure^{\oblivious}(\genericPath)\log \probabilityMeasure_{\privacyLevel}(\privatePath = \genericPath|\genericPath) \right)\label{eq:expectation_of_true-word}
    \\
    & =\exp\left( \expectation_{\genericPath \sim \probabilityMeasure^{\oblivious}}\left[\log \probabilityMeasure_{\privacyLevel}(\privatePath = \genericPath|\genericPath) \right]\right).
\end{align}
Using Lemma \ref{lemma:expectedlogproboftrueword}, we get
\begin{align}
    \mdpValue^{\private} &\geq \mdpValue^{\oblivious}-1+\exp\left( \expectation_{\genericPath \sim \probabilityMeasure^{\oblivious}}\left[\log \probabilityMeasure_{\privacyLevel}(\privatePath = \genericPath|\genericPath) \right]\right)
    \\
    &\geq \mdpValue^{\oblivious}-1+\left(\left(\outdegree_{\max}-1\right)\exp(-\frac{\epsilon}{\adjParam})+1\right)^{\numAgents \expectedLength^{\oblivious}},
\end{align}
which completes the proof.

\end{proof}

Compared to \eqref{eq:thm_1_eq}, \eqref{eq:thm_2_equation} does not take the total correlation $\totalCorrelation_{\jointPolicy}$ into account and only focuses on the success probability when the private state trajectories are the same with the original state trajectories. As a result, a joint policy $\jointPolicy = \lbrace \localPolicy^{i} \rbrace_{i=1}^{\numAgents}$ synthesized by minimizing the lower bound in \eqref{eq:thm_2_equation} does not enjoy the robustness brought by minimizing \eqref{eq:thm_1_eq}. The inclusion of total correlation in the objective function increases the team performance under private communications since the agents' policies are less sensitive to each other's state trajectories. 

\section{Details on the Independence Assumption for Local Policies} \label{sec:indepencence-appendix}
In this work, we assume that the local policies of the agents are independent from each other given the joint state. This assumption can be enforced during the synthesis procedure with the following constraint, 
\begin{equation}
    \sum_{i=1}^{\numAgents} \entropy(\mdpActionRandomVar_t^i|\gameStateRandomVar_{t}) = \entropy(\gameActionRandomVar_t^i|\gameStateRandomVar_{t}) \text{ for all } t=0,1, \ldots, \randomReachTime.
    \label{cons:localentequaljoint}
\end{equation} 
The constraint implies that the action distributions of the agents are independent given the joint state. Due to the stationarity of the policies, we can rewrite \eqref{cons:localentequaljoint} as \[ \sum_{i =1}^{\numAgents} \sum_{\mdpAction^{i}\in\mdpActionSpace^{i}}\occupancyVar_{\gameState,\mdpAction^{i}}\log\left(\frac{\sum\limits_{\mdpActionAlt^{i}\in\mdpActionSpace^{i}}\occupancyVar_{\gameState,\mdpActionAlt^{i}}}{\occupancyVar_{\gameState,\mdpAction^{i}}}\right) = \sum_{\gameAction\in\gameActionSpace}\occupancyVar_{\gameState,\gameAction}\log\left(\frac{\sum\limits_{\gameActionAlt\in\gameActionSpace}\occupancyVar_{\gameState,\gameActionAlt}}{\occupancyVar_{\gameState,\gameAction}}\right) \] for all \(\gameState \in \gameStateSpace.\) We note that both sides of the equality are concave functions of the occupancy measure variables. Similar to the objective function of that we consider, we can employ the convex-concave procedure to handle this constraint.

\end{document}